\documentclass[10pt,journal,twoside,twocolumn]{IEEEtran}
\IEEEoverridecommandlockouts
\IEEEaftertitletext{\vspace{-2\baselineskip}}
\usepackage{hyperref}
\usepackage{tikz}
\usepackage{amsgen,amsfonts,amsbsy,amsthm}
\usepackage{latexsym}
\usepackage{times}
\usepackage{array}
\usepackage{url}
\usepackage{lettrine}
\usepackage{graphicx}
\usepackage{subcaption}
\usepackage[section]{placeins}
\newtheorem{lem}{Lemma}[section]
\newtheorem{thm}{Theorem}[section]
\newtheorem{cor}{Corollary}[section]
\usepackage{amsmath}
\usepackage{amsthm}
\usepackage{amsfonts}
\usepackage{amssymb}
\usepackage{bm}
\usepackage{microtype}
\usepackage{mathrsfs}
\usepackage{graphicx}
\usepackage{cite}
\usepackage{epstopdf}
\usepackage{bbm}
\usepackage{algorithm}
\usepackage[noend]{algpseudocode}
\makeatletter
\def\BState{\State\hskip-\ALG@thistlm}
\makeatother
\DeclareMathOperator*{\argmin}{arg\,min}
\DeclareMathOperator*{\argmax}{arg\,max}
\DeclareRobustCommand*{\IEEEauthorrefmark}[1]{%
\raisebox{0pt}[0pt][0pt]{\textsuperscript{\footnotesize\ensuremath{#1}}}}
\newcommand\independent{\protect\mathpalette{\protect\independenT}{\perp}}
\def\independenT#1#2{\mathrel{\rlap{$#1#2$}\mkern2mu{#1#2}}}
\newcommand{\spn}[1]{\mbox{\texttt{span}}\left\{#1\right\}}
\newcommand{\inprod}[2]{\left\langle #1,#2 \right\rangle}
\newcommand{\bvec}[1]{\mathbf{#1}}
\newcommand{\real}{\mathbb{R}}
\newcommand{\support}{\mathcal{H}}
\newcommand{\proj}[1]{\mathbf{P}_{#1}}
\newcommand{\dualproj}[1]{\mathbf{P}_{#1}^\perp}
\newcommand{\norm}[1]{\left\|#1\right\|_2}
\newcommand{\opnorm}[1]{\left\|#1\right\|}
\newcommand{\abs}[1]{\left|#1\right|}
\newcommand{\expect}[1]{\mathbb{E}\left(#1\right)}
\newtheorem{lemma}{Lemma}

\newtheorem{definition}{Definition}
\newtheorem{remark}{Remark}

\title{Signal Recovery in Uncorrelated and Correlated Dictionaries Using Orthogonal Least Squares  }

\begin{document}

\author{\IEEEauthorblockN{Samrat~Mukhopadhyay\IEEEauthorrefmark{1}, Prateek Vashishtha\IEEEauthorrefmark{2} and Mrityunjoy Chakraborty\IEEEauthorrefmark{3}} \\
  \IEEEauthorblockA{Department of Electronics and Electrical Communication Engineering,\\ Indian Institute of Technology, Kharagpur, INDIA\\
    Email: \IEEEauthorrefmark{1}samratphysics@gmail.com,\IEEEauthorrefmark{2}prateek.iit.kgp.12@gmail.com,
    \IEEEauthorrefmark{3}mrityun@ece.iitkgp.ernet.in}}
\IEEEoverridecommandlockouts    
\maketitle
\begin{abstract}
 Though the method of least squares has been used for a long time in solving signal processing problems,
 in the recent field of sparse recovery from compressed measurements, this method has not been given much attention. In this paper we show that a method in the least squares family, known in the literature as Orthogonal Least Squares (OLS), adapted for compressed recovery problems, has competitive recovery performance and computation complexity, that makes it a suitable alternative to popular greedy methods like Orthogonal Matching Pursuit (OMP). We show that with a slight modification, OLS can exactly recover a $K$-sparse signal, embedded in an $N$ dimensional space ($K<<N$) in $M=\mathcal{O}(K\log (N/K))$ no of measurements with Gaussian dictionaries. We also show that OLS can be easily implemented in such a way that it requires $\mathcal{O}(KMN)$ no of floating point operations similar to that of OMP. In this paper  performance of OLS is also studied with sensing matrices with correlated dictionary, in which algorithms like OMP does not exhibit good recovery performance. We study the recovery performance of OLS in a specific dictionary called \emph{generalized hybrid dictionary}, which is shown to be a correlated dictionary, and show numerically that OLS has is far superior to OMP in these kind of dictionaries in terms of recovery performance. Finally we provide analytical justifications that corroborate the findings in the numerical illustrations. 
\end{abstract}

\section{Introduction}

\lettrine[findent=2pt]{\textbf{C}}{}OMPRESSED SENSING (CS) \cite{eldar2012compressed} has led to a new paradigm in signal processing. Compressed sensing provides a novel way of acquiring a sparse signal $\bm{x} \in \mathbb{R}^{N}$ such that $\|x\|_{0} \le K$ with very few number of linear measurements $M$ as compared to the original length of the signal N. The linear measurements $\bm{y} \in \mathbb{R}^{M}$ is acquired using a measurement matrix $\mathbf{\Phi}$ as.
\begin{equation}
\bvec{y}=\mathbf{\Phi x}
\end{equation}
The original signal $x$ is recovered back using a reconstruction algorithm. Compressive sensing research, is mainly concentrated around the following questions

\begin{itemize}
\item What class of measurement matrix $\mathbf{\Phi}$ can be used to acquire a compressed signal.
\item Given, the measurements $\bm{y}$ and $\mathbf{\Phi}$, which algorithms can be used to recover the original signal $\bm{x}\in R^{N}$.
\item How many measurements are required to reliably recover the signal.
\end{itemize}

Earlier, from the works of Candes-Tao~\cite{candes2006robust} and Rudelson–Vershynin~\cite{rudelson2008sparse}, it has been established that it is possible to reconstruct every $K$-sparse signal from Gaussian measurements $M$ with probability exceeding $1-e^{-cM}$ given $M \ge CK\ln(\frac{N}{K})$. The recovery is possible through the following convex minimization program.
\begin{equation}
\label{convex_program}
\min \|\bm{x}\|_{1} \hspace{0.2cm} \textnormal{s.t.} \hspace{0.2cm}
\bm{y}=\mathbf{\Phi}\bm{x}
\end{equation}
There are mainly two broad classes of algorithms which are talked about in literature. These are convex relaxation algorithms such as Basis Pursuit \cite{mallat_matching_1993} and second class of algorithms are iterative greedy algorithms~\cite{tropp2004greed}. Basis pursuit~\ref{convex_program} is an example of a convex programming approach.

However, convex program algorithms are computationally very expensive. For instance Basis Pursuit (BP) requires running time of the order of $\mathcal{O}(N^2 M^{3/2})$. As a consequence, there has been a lot of study devoted to alternative algorithms based on greedy approaches. Mallat~\cite{mallat_matching_1993} was the first to propose matching pursuit and Pati et. al.~\cite{pati-etal-1993-omp} proposed an extension of that known as Orthogonal Matching Pursuit(OMP) ~\cite{davis1997adaptive} which was one of the first of these greedy algorithms. Another early greedy algorithm is Orthogonal Least Squares(OLS)~\cite{chen1989orthogonal},\cite{natarajan1995sparse}. These algorithms greedily select indices at each step of the algorithm and append them to the already constructed support to create a successively increasing support over which they take projections to get reconstructed signal. The only difference between these two algorithms is in the method of selecting an index in the identification step. Tropp \cite{tropp2007signal} was first to find that OMP requires $\mathcal{O}(K\ln(N))$ number of measurements which is quite competitive to BP but still is quite higher than the number of measurements required for BP. However, similar efforts does not seem to have spent on the analysis of OLS. Also, the structure of OLS apparently makes it computationally more expensive, which is why in literature OLS has not not gained as much popularity as OMP.

Quite recently, Soussen et. al.~\cite{soussen2013joint} have discussed superior recovery performance of OLS in coherent dictionaries, thorough numerical simulations. Coherent dictionaries are dictionaries that have high mutual coherence~\cite{tropp2004greed}, and as a result, the columns are highly correlated. Soussen et.al. have studied OMP and OLS in these kind of dictionaries, and have given theoretical conditions for their success. However, those conditions do not seem to explain the observed superiority of OLS in recovery performance. 

\subsection{Main Objectives} It is the goal of this paper to analyze and discuss properties of OLS and try to justify the superiority it shows relative to OMP.
\begin{itemize}
\item In the first part of the paper we establish recovery guarantees for OLS under the conditions where signal is measured through a measurement matrix whose entries are i.i.d. Gaussian. Specifically, we show that a slight modification to OLS can allow it recover a $K$-sparse unknown vector with $\mathcal{O}(K\ln(N/K))$ number of measurements. 
\item Though, OLS seems to be a computationally heavy algorithm, We show that running time complexity of OLS can be made comparable to that of OMP i.e. $\mathcal{O}(KMN)$.
\item We empirically show that with correlated Measurement matrix OLS is able to successfully recover a true support set from while OMP does not.
\item Apart from empirical results, we provide analytical arguments that explain why OLS can outperform OMP in correlated dictionaries.
\end{itemize}


\section{Description of OLS for CS}
\subsection{Notation}
\label{sec:notation}
The following notation will be used throughout the paper. 
\begin{description}
\item[$\opnorm{\cdot}_p$:] The $l^p$ norm of a vector $\bvec{v}$, i.e., $\opnorm{\bvec{v}}_p=\left(\sum_{i=1}^n |v_i|^p\right)^{1/p}$.
\item[$\inprod{\cdot}{\cdot}$:] The inner product function defined on $\real^N$, defined as $\inprod{\bvec{u}}{\bvec{v}}=\sum_{i=1}^N u_iv_i,\ \forall \bvec{u,\ v}\in \real^N$.
\item[$\bvec{\Phi}$:] The real measurement matrix with $M$ rows and $N$ columns with $M<N$.
\item[$\bm{\phi}_i$:] The $i$ th column of $\bvec{\Phi}$, for $i=1,2,\cdots,\ N$. We assume that $\norm{\bm{\phi}_i}=1,\ \forall i$.
\item[$K$:]  The sparsity of unknown signal. It is assumed to be \emph{exactly} known, i.e. $\opnorm{\bvec{x}}_0=K$.
\item[$\bvec{x}_S$:] the vector $\bvec{x}$ restricted to the subset of indices $S$, i.e. ${x}_{S,i}=x_i\cdot I(i\in S)$, where $I(\cdot)$ is the set indicator function.
\item[$\support$:] The set of all the indices $\{1,2,\cdots,\ N\}$.
\item[$T$:] The unknown support set of the unknown vector $\bvec{x}$.
\item[$\bvec{\Phi}_S$:] The submatrix of $\bvec{\Phi}$ formed with the columns restricted to index set $S$.
\item[$\bvec{\Phi}_{S}^\dagger:$] The Moore-Penrose pseudo-inverse of $\bvec{\Phi}_{S}$, which exists when $\bvec{\Phi}_S$ has column rank $|S|$. It is defined as $(\bvec{\Phi}_{S}^T\bvec{\Phi}_{S})^{-1}\bvec{\Phi}_{S}^T$.
\item[$\proj{S}$:] The projection operator on $span(\bvec{\Phi}_{S})$. It is defined as  $\bvec{\Phi}_{S}\bvec{\Phi}_{S}^\dagger$.
\item[$\dualproj{S}$:] The projection operator on the orthogonal complement of $span(\bvec{\Phi}_{S})$. It is defined as $\bvec{I}-\proj{S}$.
\end{description}  
\subsection{OLS algorithm}
\label{sec:ols-algo}
\begin{table}[ht!]
\begin{subfigure}{0.5\textwidth}
\caption{OLS \textsc{Algorithm}}
\centering
\begin{tabular}{p{8.5cm}}
\hrulefill
\begin{description}
\item[\textbf{Input:}]\  measurement vector $\bvec{y}\in \real^M$, sensing matrix $\bvec{\Phi}\in \real^{M\times N}$, sparsity level $K$
\item[\textbf{Initialize:}]$\quad$ counter $k=0$, residue $\bvec{r}^0=\bvec{y}$, estimated support set, $T^0=\emptyset$, total set $\mathcal{H}=\{1,2,\cdots,\ N\}$, tolerance $\epsilon>0$ 
\item[\textbf{While}]($\norm{\bvec{r}^k}\ge  \epsilon\ \mbox{and}\ \ k<K$)
\begin{description}
\item[]\  $k=k+1$
\item[]\  {\emph{Identify:}} $\displaystyle h^k=\argmin_{i\in {\mathcal{H}}}\|\dualproj{T^{k-1}\cup \{i\}}\bvec{y}\|_2^2$
\item[]\  {\emph{Augment:}} $T^k=T^{k-1}\cup h^k$
\item[]\  {\emph{Estimate:}} $\displaystyle \bvec{x}^k=\argmin_{\bvec{u}:\bvec{u}\in \real^n,\ supp(\bvec{u})= T^k}\|\bvec{y}-\bvec{\Phi}\bvec{u}\|_2$
\item[]\  \emph{Update:} $\bvec{r}^k=\bvec{y}-\bvec{\Phi}\bvec{x}^k$
\end{description}
\item[\textbf{End While}]
\end{description}
\hrulefill
\begin{description}
\item[\textbf{Output:}]$\quad$  estimated support set $\displaystyle \hat{T}=\argmin_{S:|S|=K}\|\bvec{x}^k-\bvec{x}^k_S\|_2$ and $K$-sparse signal $\hat{\bvec{x}}$ satisfying $\hat{\bvec{x}}_{\hat{T}}=\bvec{x}^k_{\hat{T}},\ \hat{\bvec{x}}_{\support\setminus\hat{T}}=\mathbf{0}$
\end{description}
\hrulefill
\label{tab:OLS}
\end{tabular}
\end{subfigure}
\begin{subfigure}{0.5\textwidth}
\caption{OMP \textsc{Algorithm}}
\centering
\begin{tabular}{p{8.5cm}}
\hrulefill
\begin{description}
\item[\textbf{Input:}]\  measurement vector $\bvec{y}\in \real^M$, sensing matrix $\bvec{\Phi}\in \real^{M\times N}$, sparsity level $K$
\item[\textbf{Initialize:}]$\quad$ counter $k=0$, residue $\bvec{r}^0=\bvec{y}$, estimated support set, $T^0=\emptyset$, total set $\mathcal{H}=\{1,2,\cdots,\ N\}$, tolerance $\epsilon>0$ 
\item[\textbf{While}]($\norm{\bvec{r}^k}\ge  \epsilon\ \mbox{and}\ \ k<K$)
\begin{description}
\item[]\  $k=k+1$
\item[]\  {\emph{Identify:}} $\displaystyle h^k=\argmax_{i\in {\mathcal{H}}}\abs{\inprod{\bm{\phi}_i}{\bvec{r}^{k-1}}}$
\item[]\  {\emph{Augment:}} $T^k=T^{k-1}\cup h^k$
\item[]\  {\emph{Estimate:}} $\displaystyle \bvec{x}^k=\argmin_{\bvec{u}:\bvec{u}\in \real^n,\ supp(\bvec{u})= T^k}\|\bvec{y}-\bvec{\Phi}\bvec{u}\|_2$
\item[]\  \emph{Update:} $\bvec{r}^k=\bvec{y}-\bvec{\Phi}\bvec{x}^k$
\end{description}
\item[\textbf{End While}]
\end{description}
\hrulefill
\begin{description}
\item[\textbf{Output:}]$\quad$  estimated support set $\displaystyle \hat{T}=\argmin_{S:|S|=K}\|\bvec{x}^k-\bvec{x}^k_S\|_2$ and $K$-sparse signal $\hat{\bvec{x}}$ satisfying $\hat{\bvec{x}}_{\hat{T}}=\bvec{x}^k_{\hat{T}},\ \hat{\bvec{x}}_{\support\setminus\hat{T}}=\mathbf{0}$
\end{description}
\hrulefill
\label{tab:OMP}
\end{tabular}
\end{subfigure}
\end{table}
In Table.~\ref{tab:OLS} descriptions of OLS as well as the OMP algorithm are presented. It can be observed from these descriptions that OLS functionally differs from OMP only at the atom identification step. At this step, OMP chooses a new atom by evaluating the list of absolute correlations of the atoms of the dictionary with the residual vector at the last step, and then finding the index corresponding to the maximum of that list. OLS, on the other hand, at the identification step, creates a list of residual vector norm, that would have been obtained if index of a new atom of the dictionary were added to the support. OLS, then looks up that index, the inclusion of which results in the least residual vector norm. This procedure, however, seems to be formidable to work with, because it needs to evaluate the orthogonal projection error with respect to different subspaces, as indicated by the term $\norm{\dualproj{T^{k-1}\cup \{i\}}\bvec{y}}$. Fortunately, the following lemma exhibits that this selection procedure has an equivalent form that is much easier to work with.
\begin{lemma}
\label{lem:OLS-selection-step}
Let $k\ge 1$. Let $T^{k-1}$ be the support set constructed by OLS after the $(k-1)^{th}$ iteration, and  let $\bm{r^{k-1}}$ be the corresponding residual. Then, an index $i\in \support\setminus T^{k-1}$ will be chosen at the $k^{th}$ iteration if $$\bm{\phi}_{i}=\argmax_{i\in \support\setminus T^{k-1}}\frac{\abs{\inprod{\bm{\phi}_i}{\bvec{r}^{k-1}}}}{\norm{\dualproj{T^{k-1}}\bm{\phi}_i}}$$.
\end{lemma}

\begin{proof}
First, observe that, by definition, for any $i\in T^{k-1}$, $\proj{T^{k-1}\cup \{i\}}\bvec{y}\in span(\bvec{\Phi}_{T{k-1}})\implies \norm{\dualproj{T^{k-1}\cup \{i\}}\bvec{y}}=0$. Now, note that, for any $i\in T^{k-1}$,
\begin{align}
{\proj{T^{k-1}\cup \{i\}}\bvec{y}} & = \proj{T^{k-1}}\bvec{y}+\frac{\inprod{\bm{\phi}_i}{\bvec{r}^{k-1}}}{\norm{\dualproj{T^{k-1}}\bvec{\phi}_i}^2}{\dualproj{T^{k-1}}\bvec{\phi}_i}\nonumber\\
\implies {\dualproj{T^{k-1}\cup \{i\}}\bvec{y}} & =\bvec{r}^{k-1}-\frac{\inprod{\bm{\phi}_i}{\bvec{r}^{k-1}}}{\norm{\dualproj{T^{k-1}}\bvec{\phi}_i}^2}{\dualproj{T^{k-1}}\bvec{\phi}_i}\nonumber\\
\implies \norm{\dualproj{T^{k-1}\cup \{i\}}\bvec{y}}^2 & =\norm{\bvec{r}^{k-1}}^2-\frac{\abs{\inprod{\bm{\phi}_i}{\bvec{r}^{k-1}}}^2}{\norm{\dualproj{T^{k-1}}\bvec{\phi}_i}^2}
\label{eq:OLS-selection-step-equivalence}
\end{align}
Since $i\in \support\setminus T^{k-1}$ is chosen if the L.H.S. of Eq.~\eqref{eq:OLS-selection-step-equivalence} is minimized, the R.H.S. of Eq.~\eqref{eq:OLS-selection-step-equivalence} implies the desired result.
\end{proof}

\subsection{Implementation and Time Complexity of the Algorithm}
\label{sec:OLS-runtime-complexity}
Though the atom selection criteria leveraging the result of Lemma.~\ref{lem:OLS-selection-step} is relatively simpler than the original atom selection criteria of OLS, as described in Table.~\ref{tab:OLS}, it is still, apparently, seems very expensive to be implemented efficiently, because of the involvement of the orthogonal projection operators. However, exploiting QR decomposition of the projected matrices,the OLS algorithm can be implemented with a time complexity of $\mathcal{O}(KMN)$ which is same as that of OMP. This can be done by allowing twice as much space as that of OMP, maintaining the space complexity of $\mathcal{O}(MN)$. The algorithm mainly consists of two steps 
\begin{itemize}
\item Identification of the columns of the support set.
\item Computation of the signal vector $\bm{x}$
\end{itemize}
Throughout the algorithm,for any iteration say $k$ we maintain a matrix consisting the columns $\{\dualproj{T^k}\bm{\phi_{i}}\}_{i=1,2..N}$. If a column $\bm{\phi_{j}}$ is chosen at iteration $k+1$, we can modify the columns as 
\begin{equation}
\dualproj{T^{k+1}}\bm{\phi}_{i}=\dualproj{T^{k}}\bm{\phi}_{i}-\frac{\inprod{\dualproj{T^{k}}\bm{\phi}_i}{\dualproj{T^{k}}\bm{\phi}_j}}{\norm{\dualproj{T^k}\bm{\phi}_j}^2}\dualproj{T^{k}}\bm{\phi}_j
\end{equation}
Overall for $N$ columns, at any iteration the above step does not take more than $\mathcal{O}(MN)$ of floating point operations. The rest of the steps are similar to that as that of the OMP. Since, the algorithm runs for $K$ number of iterations, therefore the time complexity for identification step is $\mathcal{O}(KMN)$.\\
For,the second step, we maintain $QR$ decomposition of the selected columns. The signal vector $\bm{x}$ can then be found in not more than $\mathcal{O}(K^{2}M)$ operations.\\
Thus, the overall time complexity is $\mathcal{O}(KMN)$. 
\section{Random Measurement Ensembles}
\label{sec:random-measurement-ensemble}
In this section, we describe the type of matrix ensembles that will be used throughout the paper to carry out analysis of OLS. In this paper, two types of matrix ensembles are considered, which are refereed to ``uncorrelated'' and ``correlated'' dictionaries. These dictionaries have the following key properties:
\paragraph{\textbf{Uncorrelated dictionary}}
  This kind of matrix ensemble represents incoherent dictionaries (referred to as ``uncorrelated'' dictionaries in the sequel), i.e. dictionaries which constitute ``almost'' orthonormal matrix ensembles. Matrices of this type are assumed to have the following properties:

\begin{itemize}
\item  The entries of the matrix are i.i.d. Gaussian, i.e. $~\mathcal{N}(0,1/M)$.
\item  The columns of the matrix are stochastically independent.
\item  The columns are normalized, i.e., $\mathbb{E}\| \bm{\phi}_j \|_{2}^{2}=1,\ j=1,2,3....N$
\end{itemize} 
\paragraph{\textbf{Correlated dictionary:}}
The main property that distinguishes a correlated dictionary to a uncorrelated dictionary, is the relatively higher mutual inner product between the columns of the matrices, compared to the uncorrelated dictionaries.  One of the key measures of correlatedness of a matrix is the \emph{worst-case coherence}, defined as \begin{align*}
\mu:=\max_{i\ne j}\left\{\frac{\abs{\inprod{\bm{\phi}_i}{\bm{\phi}_j}}}{\norm{\bm{\phi}_i}\norm{\bm{\phi}_j}}\bigg|1\le i,j\le N\right\}
\end{align*}
We define the correlated dictionaries with the key property that it has high worst case coherence. For the purpose of demonstration, in this paper, a particular type of correlated dictionary is considered, that, in the sequel, will be referred to as the \emph{generalized hybrid dictionary}.
We define the generalized hybrid dictionary model as below:
\begin{definition}
\label{defn:hybrid}
A generalized hybrid dictionary of order $r$, is defined as the collection of normalized vectors $\{\bm{\phi}_{i}/\norm{\bm{\phi}_i}\}_{i=1}^N,\ \bm{\phi}_i\in\real^M $ such that \begin{align*}
\bm{\phi}_i=\bvec{n}_i+\sum_{j=1}^r u_{ij} \bvec{a}_j
\end{align*} where $\{\bvec{n}_i\}_{i=1}^N$ are i.i.d.$\sim\mathcal{N}_{M}(\bvec{0},M^{-1}I_M)$, $\{u_{ij}\}_{\scriptsize{1\le i\le N,\ 1\le j\le r}}$ are i.i.d.$\sim \mathcal{U}[0,T)$, $u_{ij}\independent \bvec{n}_k\ \forall i,j,k$ with $T>0$, and $\{\bvec{a}_i\}_{i=1}^r,\ \bvec{a}_i\in\real^M$, is an orthonormal basis for $\real^r$. 
\end{definition}
Note that this model can describe a matrix with rank $r$ in the case when $MT>>1$. Denote $W_r=\spn{\bvec{a}_1,\cdots,\ \bvec{a}_r}$. Then, the model can be seen to describe each column $\bm{\phi}_i$ as a random vector concentrated around a vector that lies inside the positive orthant of the space $W_r$, with random components, lying within the $r-$hypercube of edge length $T$, embedded in $W_r$. Specific properties of this dictionary are discussed in greater detail in Section.~\ref{sec:ols-hybrid-dictionary}.

\section{Properties of Orthonormally Projected Uncorrelated Dictionaries}
\label{sec:projected-uncorrelated-dictionary-properties}
A crucial property of OLS is that, OLS can be thought of acting like OMP at each step, but with a matrix with columns projected on a subspace. This property was exhibited by lemma.~\ref{lem:OLS-selection-step} where it was found that the $(k+1)^{\mathrm{th}}$ step of OLS can be thought as the $k^{th}$ step of OMP, but with the matrix ensemble, consisting of the columns $\{\bvec{c}_i^k\}_{i}$, where \begin{align*}
\bvec{c}_i^k=\left\{\begin{array}{ll}
\frac{\dualproj{T^k}\bvec{\phi}_i}{\norm{\dualproj{T^k}\bvec{\phi}_i}}, & i\in T^k\\
\bvec{0}, & i\notin T^k
\end{array}\right.
\end{align*} 
So, it is important to study the properties of this kind of matrix ensembles, before attempting an analysis of OLS. 
\subsection{Joint Correlation}
The following lemma shows that with high probability, a projected column of the form $\frac{\dualproj{T^k}\bvec{\phi}_i}{\norm{\dualproj{T^k}\bvec{\phi}_i}}$ is ``almost'' orthogonal to at least one of any sequence of unit norm vectors. For the uncorrelated dictionaries that are considered here with i.i.d. gaussian entries, one can use standard concentration inequalities, along with some results from matrix theory to establish this result as shown below.
\begin{lem}
\label{lem:projected-joint-correlation}
Consider the uncorrelated dictionary $\bm{\Phi}$ assumed in Section.~\ref{sec:random-measurement-ensemble}. Let ${\bm{u}}\in \real^{M}$ be a vector whose $l_{2}$ norm do not exceed one. Given set $T^k$ of columns of $\bm{\Phi}$ and let $\bm{z}\in \real^{M}$ be a gaussian random vector with i.i.d. entries, independent of ${\bm{u}}$ and the columns of $\bm{\Phi}_{T^k}$. Then,
\begin{align}
\mathbb{P}\left\{\abs{\inprod{\frac{\dualproj{T^k}\bvec{z}}{\norm{\dualproj{T^k}\bvec{z}}}}{\bvec{u}} }\le \epsilon\right\} & \ge 1-e^{-\frac{m(M_1-1)\epsilon^2}{2}}
\end{align} 
where $M_{1}=M-k-1$, and $m(n)=(\sqrt{n-1}+1)^2,\ \forall n\in \mathbb{N}$. 
\end{lem}

\begin{proof}
Since $\bvec{z},\bvec{u}\independent \bvec{\Phi}_{T^k}$, we first find a lower bound of the probability of the desired event, conditioned on the fact that the columns of $\bvec{\Phi}_{T^k}$, and the vector $\bvec{u}$ are given. 

Now, given the $\bvec{\Phi}_{T^k}$, note that $\dualproj{T^k}$ is an orthogonal projection error operator and hence can be decomposed as \begin{align*}
\dualproj{T^k}=\bvec{U\Sigma U}^T
\end{align*} where $\bvec{U}\in \real^{M\times M}$ is an orthogonal matrix, and $\Sigma$ is a diagonal matrix, with, w.l.o.g. its first $M-d$ diagonal elements $1$ and the rest $0$, where $d=\dim R(\bm{\Phi}_{T^k})\le k$, so that $d\le k$. Then, one can write, \begin{align*}
\lefteqn{\inprod{\frac{\dualproj{T^k}\bvec{z}}{\norm{\dualproj{T^k}\bvec{z}}}}{\bvec{u}}} & &\\
\ & =\inprod{\frac{\bvec{U\Sigma U}^T\bvec{z}}{\norm{\bvec{U\Sigma U}^T\bvec{z}}}}{\bvec{u}} \\
\ & =\inprod{\frac{\bvec{\Sigma U}^T\bvec{z}}{\norm{\bvec{\Sigma U}^T\bvec{z}}}}{\bvec{\Sigma U}^T\bvec{u}}\\
\ & =\inprod{\frac{\bvec{v}}{\norm{\bvec{v}}}}{\bvec{\tilde{u}}}
\end{align*}
where \begin{itemize}
\item $\bvec{v}\in \real^{M-d}$, with its components as the first $M-d$ components of $\bvec{U}^T\bvec{z}$
\item $\bvec{\tilde{u}}\in \real^{M-d}$, with its components as the first $M-d$ components of $\bvec{U}^T\bvec{u}$
\end{itemize}
Two further observation are in order:
\begin{itemize}
\item Given $\bvec{\Phi}_{T^k}$, $\bvec{U}^T\bvec{z}\sim \mathcal{N}(\bvec{0},M^{-1}\bvec{I}_M)$\footnote{Since given $\bvec{U}$, by independence of $\bvec{\Phi}_{T^k}$ and $\bvec{x},\ \expect {(\bvec{U}^T\bvec{z})(\bvec{U}^T\bvec{z})^T\mid \bvec{U}}=\bvec{U}^T\expect{\bvec{z z}^T}\bvec{U}= \bvec{U}^T M^{-1}\bvec{I}_M\bvec{U}=M^{-1}\bvec{I}_M$}, which implies that $\bvec{v}\sim \mathcal{N}(\bvec{0},M^{-1}\bvec{I}_{M-d})$. 
\item $\norm{\bvec{\tilde{u}}}=\norm{\Sigma\bvec{U}^T\bvec{u}}\le\norm{\bvec{u}}\le 1$.
\end{itemize}  

A further simplification can be furnished by finding an orthogonal matrix $\bvec{U}_1\in \real^{(M-d)\times (M-d)}$, i.e. a rotation in the $(M-d)$ dimensional space that transforms $\bvec{\tilde{u}}$ to a vector lying on one of the coordinate axes; specifically, $\bvec{\hat{u}}:=\bvec{U}_1\bvec{\tilde{u}}$ has its first component nonzero and all the other components $0$. This task can be executed by constructing $\bvec{U}_1$ by putting in the first row the vector $\bvec{\tilde{u}}/\norm{\bvec{\tilde{u}}}$, and putting in the rest of the rows an orthonormal basis of the orthogonal complement of $\spn{\bvec{\tilde{u}}}$ in $\real^{M-d}$. Then, one can write, \begin{align*}
\lefteqn{\inprod{\frac{\bvec{v}}{\norm{\bvec{v}}}}{\bvec{\tilde{u}}}} & & \\
\ =& \inprod{\frac{\bvec{v}_1}{\norm{\bvec{v}_1}}}{\bvec{\hat{u}}}
\end{align*}
where $\bvec{v}_1=\bvec{U}_1\bvec{v}$. Note that, since $\bvec{U}_1$ is constructed from $\bvec{\tilde{u}}$ which is independent of $\bvec{v}$, conditioned on $\{\bvec{u}\}$ and $\bvec{\Phi}_{T^k}$, $\bvec{v}_1\sim \mathcal{N}(\bvec{0},M^{-1}\bvec{I}_{M-d})$.  

At this point, it is useful to make a change of coordinates from Cartesian to polar, to represent $\frac{\bvec{v}_1}{\norm{\bvec{v}_1}}$ as \begin{align*}
\begin{bmatrix}
\cos\Theta_1\\
\sin \Theta_1\cos \Theta_2\\
\vdots\\
\sin\Theta_1\sin\Theta_2\cdots\sin\Theta_{M_1-1}\cos\Theta_{M_1}\\
\sin\Theta_1\sin\Theta_2\cdots\sin\Theta_{M_1-1}\sin\Theta_{M_1}
\end{bmatrix}
\end{align*}
Where $M_1=M-d-1$. Here, given $\bvec{\Phi}_{T^k}$, $\Theta_1, \Theta_2,\ \cdots,\ \Theta_{M_1}$ are independent, but not identically distributed, continuous valued random variables, with $\Theta_1,\ \cdots,\ \Theta_{M_1-1}\in [0,\pi)$ and $\Theta_{M_1}\in [0,2\pi)$. The probability distribution function of these random variables are given by (conditioned on $\bvec{\Phi}_{T^k}$)\begin{align*}
p_{\Theta_i}(\theta_i)=\left\{\begin{array}{ll}
\frac{(\sin\theta_i)^{M_1-i}}{\beta\left(\frac{M_1}{2},\frac{1}{2}\right)}, & 1\le i\le M_1-1\\
\frac{1}{2\pi}, & i=M_1
\end{array}\right.
\end{align*}
Therefore, we have \begin{align*}
\lefteqn{\inprod{\frac{\dualproj{T^k}\bvec{z}}{\norm{\dualproj{T^k}\bvec{z}}}}{\bvec{u}}} & & \\
\ =& \cos\Theta_1
\end{align*}
Since $\Theta_1$ has density $p_{\Theta_1}$ when $\bvec{\Phi}_{T^k}$ is given, we find \begin{align*}
\lefteqn{\mathbb{P}\left\{\abs{\inprod{\frac{\dualproj{T^k}\bvec{z}}{\norm{\dualproj{T^k}\bvec{z}}}}{\bvec{u}}}\ge  \epsilon\mid \bvec{\Phi}_{T^k},\bvec{u}\right\}} & & \\
\ &=\mathbb{P}(\abs{\cos \Theta_1}\ge \epsilon\mid \bvec{\Phi}_{T^k},\bvec{u})\\
\ &=2\int_{0}^{\cos^{-1}\epsilon}p_{\Theta_1}(\theta_1)d\theta_1\\
\ &=\frac{2}{\beta\left(\frac{M_1}{2},\frac{1}{2}\right)}\int_{0}^{\cos^{-1}\epsilon}\sin^{M_1-1}\theta_1 d\theta_1\\
\ &=f_{M_1-1}(\epsilon^2)
\end{align*}
where, for a given $n\in \mathbb{N}$, $f_n:[0,1]\to [0,1]$ is defined as \begin{align*}
f_{n}(x)=\frac{1}{A_n}\int_{0}^{\sqrt{1-x}}\frac{u^n}{\sqrt{1-u^2}}du,\quad \forall x\in [0,1]
\end{align*}
where $A_n=\frac{1}{2}\beta\left(\frac{n+1}{2},\frac{1}{2}\right)$. The following lemma is invoked to find a upper bound of the desired probability. 
\begin{lem}
\label{lem:theta_1-prob-upper-bound}
$\forall n\in \mathbb{N},\ \forall x\in [0,1]$, \begin{align}
\label{eq:theta_1-prob-upper-bound}
f_n(x)\le e^{-\frac{m(n) x}{2}}
\end{align}
where $m(n):=(\sqrt{n-1}+1)^2$.
\end{lem}
\begin{proof}
The proof is postponed to Appendix.~\ref{sec:proof-lemma-theta_1-prob-upper-bound}
\end{proof}
Invoking Lemma.~\ref{lem:theta_1-prob-upper-bound}, it is immediate to find \begin{align*}
\lefteqn{\mathbb{P}\left\{\abs{\inprod{\frac{\dualproj{T^k}\bvec{z}}{\norm{\dualproj{T^k}\bvec{z}}}}{\bvec{u}}}\ge \epsilon\mid \bvec{\Phi}_{T^k},\bvec{u}\right\}} & & \\
\le & e^{-\frac{m(M_1-1)\epsilon^2}{2}}
\end{align*}
Thus, the desired probability can be upper bounded as \begin{align*}
\lefteqn{\mathbb{P}\left\{\abs{\inprod{\frac{\dualproj{T^k}\bvec{z}}{\norm{\dualproj{T^k}\bvec{z}}}}{\bvec{u}}}\ge \epsilon\right\}} & & \\
\ =& \int_{\bvec{\Phi}_{T^k},\ \bvec{u}_t}\mathbb{P}\left\{\abs{\inprod{\frac{\dualproj{T^k}\bvec{z}}{\norm{\dualproj{T^k}\bvec{z}}}}{\bvec{u}_t}}\ge \epsilon\mid \bvec{\Phi}_{T^k},\bvec{u}\right\}d\mathbb{P}(\bvec{\Phi}_{T^k})d\mathbb{P}(\bvec{u}_t)\\
\ \le &  e^{-\frac{m(M_1-1)\epsilon^2}{2}}\\
\ \le &  e^{-\frac{(\sqrt{M-k-2}+1)^2\epsilon^2}{2}}=e^{-\frac{m(M_1-1)\epsilon^2}{2}}
\end{align*}
where $M_1$, with a slight abuse of notation, is again defined as $M-k-1$.
\end{proof}

\subsection{Smallest Singular Value}
Inspired by the analysis technique introduced by Tropp and Gilbert in~\cite{tropp2007signal}, it is imperative to find out a tail bound of the smallest singular value of the projected uncorrelated matrix defined before in this section. Before proceeding to find the tail bound, we recall an important result associated with the tail bounds for ``unprojecetd'' uncorrelated matrices, i.e. matrices with the uncorrelated dictionary without being projected to any subspace. To do so, the unprojected matrix, is assumed to satisfy the following type of concentration inequality:   
\begin{align}
\label{eq:concentration_inequality}
\mathbb{P}\left\{ \abs{ \norm{\bvec{\Phi x}}^2-\norm{\bvec{x}}^2 } \le \epsilon \norm{\bvec{x}}^2\right\} \ge 1-2e^{-M c_{0}(\epsilon)},\quad 0 < \epsilon \le 1
\end{align}
where $c_0(\epsilon)$ is a constant that depends only on $\epsilon$, such that $c_0(\epsilon)>0\ \forall
\epsilon\in (0,1)$. Then the singular values can be bounded with high probability thanks to the following lemma due to Baraniuk , 


\begin{lem}
\label{lem:Baranuik_singular_value_bound}
Suppose that $\bm{Z}\in \real^{M\times K}$ be a sub-matrix of a $M\times N$ matrix $\bvec{\Phi}$, suhc that $\bvec{Z}$ satisfies the concentration inequality \eqref{eq:concentration_inequality}. Then, for any $\epsilon\in (0,1)$ and $\bvec{x}\in \real^K$, one has
\begin{align}
(1-\epsilon)\|\bm{x}\|_{2}\le\|\bm{Z x}\|_2\le (1+\epsilon)\|\bm{x}\|_{2} \hspace{1cm} 
\end{align}
with probability exceeding $1-2(12/\epsilon)^Ke^{-c_0(\epsilon/2)M}$.
\end{lem}
It should be emphasized that this result is true for a \emph{given} $M\times K$ sub-matrix of $\bvec{\Phi}$, as opposed to the stronger result, generally alluded to the restricted isometry property (RIP), which gives conditions ensuring the above kind of bound to hold for \emph{all} $M\times K$ sub-matrices of $\bvec{\Phi}$.
%
%
%
%

Another important result, specific to the case of Gaussian matrices, is attributed to Davidson and Zarek, which gives a much tighter estimate for lower bound of the lowest singular value. This one will be more useful in our analysis as we have considered Gaussian entries for our matrices.

\begin{align}
\label{eq:zarek-singular-bound}
\mathbb{P} \left\lbrace \sigma_{min}({\bm{Z}}) \ge 1-\sqrt{\frac{K}{M}} - \epsilon \right\rbrace \ge 1-e^{-\frac{\epsilon^{2}M}{2}}
\end{align}
In the following lemmas, we now attempt to find similar estimates of lower bounds on the lowest singular values of the ``projected'' uncorrelated matrix, as defined at the beginning of Section.~\ref{sec:projected-uncorrelated-dictionary-properties}.
\begin{lem}
\label{lem:least-singular-value-lem1}
Let $\mathbf{\Phi}$ $\in \mathbf{R^{M \times K}}$ and let $T^{k}$ be a set of $k$ arbitrary chosen indices from $\{1,2,\cdots,\ K\}$. Then, \begin{align*}
\sigma_{\mathrm{min}}(\mathbf{P_{T^{k}}^{\perp}}\bm{\Phi}_{(T^{k})^{c}})\ge \sigma_{\mathrm{min}}(\mathbf{\Phi})
\end{align*}
\end{lem}

\begin{proof}
Since the singular values of a matrix is invariant under permutation of its columns, without any loss of generality, we can partition matrix $\mathbf{\Phi}$ as  $$\begin{bmatrix}
\bvec{\Phi}_{T^k} & \bvec{\Phi}_{(T^k)^c}
\end{bmatrix}$$ Then we have 
\begin{align*}
\bvec{\Phi}^T\bvec{\Phi}=\begin{bmatrix}
\bm{\Phi}_{T^k}^T\bm{\Phi}_{T^k} & \bm{\Phi}_{T^k}^T\bm{\Phi}_{(T^k)^c}\\
\bm{\Phi}_{(T^k)^c}^T\bm{\Phi}_{T^k} & \bm{\Phi}_{(T^k)^c}^T\bm{\Phi}_{(T^k)^c}
\end{bmatrix}
\end{align*}

Now, Clearly the upper left entry of the block matrix $(\bm{\Phi}^T\bm{\Phi})^{-1}$ would be the inverse of the Schur complement of the matrix $\bm{\Phi}^{T}\bm{\Phi}$ which is clearly $(\bm{\Phi}_{(T^{k})^{c}}^{T}\mathbf{P_{T^{k}}^{\perp}}\bm{\Phi}_{(T^{k})^c})^{-1}$. Now, using Cauchy's interlacing theorem for eigenvalues of Hermitian matrices~\cite{horn2012matrix} we have $\lambda_{\mathrm{max}}(\bm{\Phi}_{(T^{k})^{c}}^{T}\mathbf{P_{T^{k}}^{\perp}}\bm{\Phi}_{(T^{k})^c})^{-1} \le \lambda_{\mathrm{max}}(\mathbf{\Phi}^T\mathbf{\Phi})^{-1}$ and thus, $\lambda_{\mathrm{min}}(\bm{\Phi}_{(T^{k})^{c}}^{T}\mathbf{P_{T^{k}}^{\perp}}\bm{\Phi}_{(T^{k})^c}) \ge \lambda_{\mathrm{min}}(\mathbf{\Phi}^T\mathbf{\Phi})$. Now recall that, for any matrix $\bvec{A}$, $\sigma_{\mathrm{min}}(\bvec{A})=\sqrt{\lambda_{\min}(\bvec{A}^T\bvec{A})}$, from which the desired result follows directly.                                              
\end{proof}

\begin{lem}
\label{lem:least-singular-value-lem2}
Let $\bvec{\Phi}\in \real^{M\times K}$ and $T^k$ represent a set of $k$ arbitrary indices from $\{1,2,\cdots,\ K\}$. Let $\bm{z}_{1}, \bm{z}_{2}, \bm{z}_{3}....\bm{z}_{L} \in R^{M}$ be random vectors i.i.d. $\mathcal{N}(0,M^{-1}\bvec{I})$ and are independent of the columns of $\bvec{\Phi}_{T^{k}}$. Then the following event holds with probability exceeding $1-2Le^{{-\frac{\delta^{2}(M-d)}{8}}}$
\begin{align}
\ & \left\lbrace\left(1-\frac{d}{M}\right) \left(1-\delta \right) \le \| \mathbf{P_{T^{k}}^{\perp}}\bm{z}_i \|_{2}^{2} \le \left(1-\frac{d}{M}\right) \left(1+\delta \right)\forall\  1\le i\le L\nonumber\right\rbrace 
\end{align}
where $d=\dim R(\bm{\Phi}_{T^k})\le k$.
\end{lem}

\begin{proof}
Recalling the notation introduced in Section.~\ref{sec:notation}, it is easy to see that $\mathbf{P_{T^{k}}^{\perp}}$ is an idempotent matrix with $M-d$ eigenvalues $1$ and $d$ eigenvalues $0$, where $d=\dim{{R}}(\bvec{\Phi}_{T^k})$. Therefore, we can decompose $\mathbf{P_{T^{k}}^{\perp}}$ as $\mathbf{U \Sigma U^{T}}$ where $\mathbf{\Sigma}$ is a diagonal matrix containing $M-d$ $1$s and $d$ $0$s, and $\mathbf{U}$ is an orthonormal matrix. \\

Now, for any $\bm{z}\in\real^{M}$, distributed as $\mathcal{N}(\bvec{0},M^{-1}\bvec{I})$, and independent of $\bvec{\Phi}_{T^k}$, the following observation can be made 
\begin{align}
{\norm{\dualproj{T^k}\bvec{z}}^2} & =\bvec{z}^T\dualproj{T^k}\bvec{z}\nonumber\\
\ &=\bvec{z}^T\bvec{U}\bvec{\Sigma}\bvec{U}^T\bvec{z}^T\nonumber\\
\ &=(\bvec{\Sigma U}^T \bvec{z})^T(\bvec{\Sigma U^T z})
\end{align}
Since $\mathbf{U}$ is an orthonormal matrix independent of $\bm{z}$, given $\bvec{\Phi}_{T^k}$, $\mathbf{U^{T}}\bm{z}$ is a random vector with entries i.i.d Gaussian distributed. Now, let $\bm{v}$ be the vector with the first $M-d$ components of $\Sigma U^{T}\bm{z}$. Then $\bm{v} \sim \mathcal{N}(\bvec{0},M^{-1}\bvec{I}_{M-d})$ (given $\bvec{\Phi}_{T^k}$). Thus, given $\bvec{\Phi}_{T^k}$, $\bm{v}$ is an i.i.d distributed Gaussian vector with ${\mathbb{E}(\|v\|_2^{2})}=\left(1-\frac{d}{M}\right)$. Now, a standard exercise in concentration inequalities show that norm of $\bvec{v}/\sqrt{1-d/M}$ will be concentrated about its mean~\cite{foucart2013mathematical}, i.e. the following holds 
\begin{align}
\mathbb{P}\left\lbrace\abs{\norm{\dualproj{T^k}\bvec{z}}^{2}-{\left(1-\frac{d}{M}\right)}}\ge {\delta\left(1-\frac{d}{M}\right)}\right\rbrace \le 2e^{-\frac{\delta^{2}(M-d)}{8}}
\end{align}
Then, considering the probability for the complementary events and taking union bound over $L$ vectors, we arrive at the desired result.
\end{proof}

\begin{lem}
\label{lem:least-singular-value-lem3}
Let $\bm{\Phi}\in R^{M\times K}$, and let $\bm{\Phi}$ satisfies smallest singular value property as in \ref{lem:Baranuik_singular_value_bound} with $\sigma_{min}\ge \sigma$. Let $T^k$ be a set of $k$ indices. Then $\|(\mathbf{P_{T^{k}}^{\perp}}\bm{\Phi}_{(T^{k})^{c}} \bvec{D})\bvec{x}\|_{2} \ge \sigma\sqrt{1-\delta}$ with probability exceeding \begin{align*}
\ & 1-2(12/(1-\sigma))^Ke^{-c_0((1-\sigma)/2)M}-2(K-k)e^{-(M-k)\delta^2/8}  \\
\ & >1-e^{-cM}
\end{align*} for some $M>CK$ where 
$$\mathbf{D}=\mathrm{diag}\left(\frac{1}{\|\mathbf{P_{T^{k}}^{\perp}}\bm{\phi}_{1}\|_{2}},\ \frac{1}{\norm{\dualproj{T^k}\bvec{\phi}_2}}\cdots,\ \frac{1}{\norm{\dualproj{T^k}\bvec{\phi}_{K-k}}}\right)$$ where $\bm{\phi}_{i} \in (T^{k})^{c}$ for all $\bvec{x} \in R^{K-k}$ with $\norm{\bvec{x}}=1$, and $k=1,2,3,\cdots,\ K-1$.
\end{lem}

\begin{proof}
First of all note that, for any $\bvec{x}\in \real^{K-k}$ such that $\norm{\bvec{x}}=1$, \begin{align*}
\norm{\dualproj{T^k}\bvec{\Phi}_{(T^k)^c}\bvec{D x}}^2 & =\bvec{x}^T\bvec{D}^T\bvec{\Phi}_{(T^k)^c}^T\dualproj{T^k}\bvec{\Phi}_{(T^k)^c}\bvec{D}\bvec{x}\\
\ &\ge \sigma_{\min}^2(\dualproj{T^k}\bvec{\Phi}_{(T^k)^c})\sigma_{\min}^2(\bvec{D})\\
\ &\ge \sigma_{\min}^2(\bvec{\Phi})\frac{1}{\max_{i\in (T^k)^c}\norm{\dualproj{T^k}\bvec{\phi}_i}^2}
\end{align*}
where the last inequality uses Lemma.~\ref{lem:least-singular-value-lem1}. Thus, \begin{align*}
\mathbb{P}\left(\norm{\dualproj{T^k}\bvec{\Phi}_{(T^k)^c}\bvec{D x}}\ge \sigma\sqrt{1-\delta}\right) & \ge \mathbb{P}(E_1\cap E_2) \\
\ \ge 1-\mathbb{P}(E_1^c)-\mathbb{P}(E_2^c)
\end{align*}
where \begin{align*}
E_1:= & \left\{\sigma_{\min}(\bvec{\Phi})\ge \sigma\right\}\\
E_2:= & \left\{{\max_{i\in (T^k)^c}\norm{\dualproj{T^k}\bvec{\phi}_i}}\le \sqrt{1-\delta}\right\}
\end{align*}
Now, if $\bm{\Phi}$ satisfies smallest singular value property as in \ref{lem:Baranuik_singular_value_bound}, with $\sigma_{min}\ge\sigma$ then clearly using \ref{lem:least-singular-value-lem1} we can conclude that matrix $\mathbf{P_{T^{k}}^{\perp}}\bm{\Phi}_{(T^{k})^{c}}$ satisfies the least singular value bound with probability atleast $1-2(12/(1-\sigma))^Ke^{-c_0((1-\sigma)/2)M}$. On the other hand, Lemma.~\ref{lem:least-singular-value-lem2} dictates that \begin{align*}
\mathbb{P}(E_2)\ge & \mathbb{P}\left\{{\max_{i\in (T^k)^c}\norm{\dualproj{T^k}\bvec{\phi}_i}}^2\le \left(1-\frac{d}{M}\right)(1-\delta)\right\}\\
\ \ge & 1-2(K-k)e^{-(M-d)\delta^2/8}\\
\ \ge & 1-2(K-k)e^{-(M-k)\delta^2/8} 
\end{align*} 
Thus, \begin{align*}
\lefteqn{\mathbb{P}\left(\norm{\dualproj{T^k}\bvec{\Phi}_{(T^k)^c}\bvec{D x}}\ge \sigma\sqrt{1-\delta}\right)} & & \\
\ & \ge 1-2(12/(1-\sigma))^Ke^{-c_0((1-\sigma)/2)M}-2(K-k)e^{-(M-k)\delta^2/8}
\end{align*} 
For Gaussian matrices, one can use $c_0(\epsilon)=\epsilon^2/8$ to further simplify the bound.

\end{proof}
\section{Analysis of OLS in Uncorrelated dictionaries}
\label{sec:ols-uncorrelated-dictionaries}
The following theorem is the one of the main results in the paper. In this theorem, we argue that a slight modification to OLS Algorithm can be shown to be requiring $\mathcal{O}(K\log(N/K))$ number of measurements for perfect recovery, which is asymptotically the same as that required by Basis Pursuit. Our modification concerns with the first iteration in OLS.
We discuss about our modification within the proof of this theorem.
\begin{thm}
\label{thm:uncorrelated-dcitionary-recovery-probability}
Let $\mathbf{\Phi} \in \mathbf{R^{M \times N}}$ be a measurement matrix and let $\mathbf{x} \in \mathbf{R^{N}}$ be an arbitrary $K$-sparse signal. Given $\mathbf{y} = \mathbf{\Phi x} $, the  measurement vector. Then, Orthogonal Least Squares algorithm can reconstruct the signal $\mathbf{x}$ with probability exceeding $1-\delta$, with $\delta \in (0,1)$, for number of measurements $M>CK\ln\left(\frac{N}{Kc(\delta)}\right)+K+1$ for some suitably chosen constant $C>0$, and some suitable chosen constant $c(\delta)$ that depends on $\delta$.
\end{thm}

\begin{proof}
Our proof of this Theorem is inspired by the approached adopted by Tropp \cite{tropp2007signal}. The main innovation in our proof relies on the fact that the first iteration of OLS is essentially the same as that of OMP, where a column is chosen which has a maximum absolute correlation with measurement $\bvec{y}$. From second iteration onwards, the selection criteria of OLS is unique to itself. Let us consider the greedy selection rule for OLS from $2^{nd}$ iteration onwards.
\begin{equation}
\begin{split}
\rho(\mathbf{r}) &= \frac{\| (\mathbf{P_{T^{k}}^{\perp} {\Psi \bvec{D}_{1}})^{T}\bm{r_{k}}} \|_{\infty}}{\| (\mathbf{P_{T^{k}}^{\perp} {\mathbf{\Phi}_{S}\bvec{D}_{2}})^{T}\bm{r_{k}}} \|_{\infty}} = \frac{\max_{\bm{\psi}} \left| \langle \frac{\mathbf{P_{T^{k}}^{\perp} \bm{\psi}}}{\|\mathbf{P_{T^{k}}^{\perp} \bm{\psi}}\|_{2}},\bm{r_{k}} \rangle \right|}{\| (\mathbf{{P_{T^{k}}^{\perp}\mathbf{\Phi}_{S}\bvec{D}_{2}})^{T}\bm{r_{k}}} \|_{\infty}} \\
D_{1}(i) &= \frac{1}{\| \mathbf{P_{T^{k}}^{\perp}} \bm{\psi}_{i} \|_{2}} \\
D_{2}(j) &= \frac{1}{\| \mathbf{P_{T^{k}}^{\perp}}  \bm{\phi}_{j} \|_{2}} \textnormal {if $j \notin T^{k}$ otherwise 0} \\
\end{split}
\end{equation}
where $D_{k}(i)$ represents $i^{th}$ entry of a diagonal matrix for $k=1,2$
with $\mathbf{r}=\mathbf{y}$ where $ \mathbf{\Phi_{S}}$ is the matrix formed by the stacking the columns indexed by the true support set $S$. At the $k^{\mathrm{th}}$ iteration,the OLS algorithm chooses a column from the true support set $S$ if and only if $\rho(\bm{r_{k}}) < 1$. This can be shown aloing the same lines of the argument that Tropp \cite{tropp2004greed} used to show how the greedy selection rule $\rho(\mathbf{r}) < 1$ ensures the reconstruction of signal $\bvec{x}$ by the OMP algorithm in $K$ iterations under non-noisy conditions. The gist of the idea for OMP is to consider an imaginary and a real execution of the OMP algorithm. Starting with initial residual as $r_{0}=y$ and posing the induction arguments, Tropp shows that the greedy selection rule ensures that the OMP algorithm will identify the correct set of indices from the true support set $S$ of signal $\bvec{x}$. In our case we pose the same arguments as above for OLS algorithm but from the $2^{nd}$ iteration onwards and with a different greedy selection rule. 

Before proceeding with the proof, we propose a slight modification to the OLS Algorithm. The purpose of this modification will eventually become clear. We choose a dummy column say $\bm{\phi}_{d}$ whose distribution is same as that of the other columns and a signal value say $\alpha$. Then, we simply modify our measurement vector $y_{1}=y+\alpha\bm{\phi}_{d}$. Consequently, we can warm start the OLS algorithm with $\bm{\phi}_{d}$ as the column preselected for the first iteration and proceed with the rest of the iterations in the usual way OLS works. The effect of this modification is that now we are able to bypass the first iteration of OLS, whose selection criteria is same as that of OMP, and instead use the greedy selection criteria unique to OLS, from second iteration onwards. We call this as the $0^{th}$ iteration.

The rest of the iterations will continue until it has picked $K$ columns. 

The two conditional events we consider are,

\begin{itemize}
\item $\Sigma ={\sigma_{min}(\mathbf{\Phi_S}) > \sigma}$
\item $E_{1}$ =Success at the first iteration    
\end{itemize}

Conditioning on the above two events, we denote $\mathbb{P}(E_{S})$ as the overall probability of success of OLS i.e. OLS is able to recover $K$-sparse signal.

\begin{equation}
E_{S} = {\rho (\bm{r_{k}}) < 1 , \hspace{0.2cm} \textnormal{for} \hspace{0.2cm} k=1,2,3,4,...K}
\end{equation}


Recall Lemma.~\ref{lem:least-singular-value-lem2} to appreciate that, for any iteration $k=1,2,,\cdots,\ KK$, the event $\Sigma$ implies that $\sigma_{\min} (\mathbf{P_{T^{k}}^{\perp}\bm{\phi}_{(T^{k})^{c}}\bvec{D}}) > \sigma\sqrt{(1-\delta)}$. We denote $\sigma_{min}=\sigma \sqrt{(1-\delta)}$ Also, clearly,
\begin{equation}
\mathbf{P_{T^{k}}^{\perp}\bm{\phi}_{(T^k)^{c}}\bvec{D}} = \mathbf{({P_{T^{k}}^{\perp}\mathbf{\Phi}_{S}\bvec{D}_{2}})}
\end{equation}

To ensure success for iterations $k = 1,2,\cdots,\ K $ with residual $\bm{r_{k}}$, we require,

\begin{equation}
\rho(\bm{r_{k}}) < 1 \implies \frac{\max_{\mathbf{\bm{\psi}}} \left| \langle \frac{\mathbf{P_{T^k}^{\perp} \bm{\psi}}}{\|\mathbf{P_{T^k}^{\perp} \bm{\psi}}\|_{2}},\bm{r_{k}} \rangle \right|}{\| (\mathbf{{P_{T^k}^{\perp}\mathbf{\Phi}_{S}D_{2}})^{T}\bm{r_{k}}} \|_{\infty}} < 1
\end{equation}

At any iteration $k$, assuming all previous iterations were successful, $\bm{r_{k}}$ lies in the span of $\mathbf{P_{T^k}^{\perp} \mathbf{\Phi}_{S}D_{2}}$, so that $\| \mathbf{P_{T^k}^{\perp} \mathbf{\Phi}_{S}D_{2}r_{k}} \|_{2} \ge \sigma_{\min} \| \bm{r_{k}}\|_{2}$
Thus, 
\begin{equation}
\| \mathbf{(P_{T^k}^{\perp} \mathbf{\Phi}_{S} D_{2})^{T}r_{k}} \|_{\infty} \ge \frac{\| \mathbf {( P_{T^k}^{\perp} \mathbf{\Phi}_{S} D_{2} )^{T} \bm{r}_{k} } \|_{2} }{\sqrt(K)} \ge \frac{ \sigma_{\min} \| \bm{r_{k}}\|_{2} }{\sqrt{K}}
\end{equation}
\begin{equation}
\begin{split}
&\mathbb{P}(\rho(\bm{r_{k}}) < 1 \hspace{0.1cm} \forall \hspace{0.1cm} k=1,2,3,4....K)\\
&\ge \mathbb{P}\left( \max_{k} \frac{\sqrt{K} \max_{\mathbf{\bm{\psi}}} \left| \langle \frac{\mathbf{P_{T^k}^{\perp} \bm{\psi}}}{\|\mathbf{P_{T^k}^{\perp} \bm{\psi}}\|_{2}},\bm{r_{k}} \rangle \right| }{ \| \mathbf {( P_{T^k}^{\perp} \mathbf{\Phi}_{S} D_{2} )^{T} r_{k} } \|_{2}} < 1 \right) \\
&\ge \mathbb{P}\left( \max_{k}   \max_{\mathbf{\bm{\psi}}} \left| \langle \frac{\mathbf{P_{T^k}^{\perp} \bm{\psi}}}{\|\mathbf{P_{T^k}^{\perp} \bm{\psi}}\|_{2}},\mathbf{\frac{r_{k}}{\| r_{k} \|_{2}}} \rangle \right| < \frac{\sigma_{\min}}{\sqrt(K)} \right)\\
\end{split}
\end{equation}



Using stochastic independence of columns of matrix $\mathbf{\Phi}$, for $\mathbb{P}(\rho(r_{k}) < 1)$ for $k=1,2,\cdots,\ K$, and defining $\bm{u_k}=\mathbf{\frac{r_{k}}{\| r_{k}\|_{2} }}$ and $\epsilon=\frac{\sigma_{\min}}{\sqrt{K}}$, we express the lower bound as 
\begin{align}
\label{eq:joint-rpoduct}
\ & \prod_{\mathbf{\bm{\psi}}}\mathbb{P}\left( \max_{k} \left| \left\langle \frac{\mathbf{P_{T^k}^{\perp} \bm{\psi}}}{\|\mathbf{P_{T^k}^{\perp} \bm{\psi}}\|_{2}},\mathbf{\frac{r_{k}}{\| r_{k} \|_{2}}} \right\rangle \right| < \epsilon \right)
\end{align}
Now, recall Lemma.~\ref{lem:projected-joint-correlation} to find that \begin{align}
\  & \mathbb{P}\left( \max_{k} \left| \left\langle \frac{\mathbf{P_{T^k}^{\perp} \bm{\psi}}}{\|\mathbf{P_{T^k}^{\perp} \bm{\psi}}\|_{2}},\mathbf{\frac{r_{k}}{\| r_{k} \|_{2}}} \right\rangle \right| < \epsilon \right)\nonumber\\
\  & \ge 1-\sum_{k=1}^K e^{\frac{-m(M-k-1)\epsilon^2}{2}}\nonumber\\
\  & =1-K\exp\left(-\frac{\sigma_{\min}^2\sqrt{M_1}}{K}\right) e^{-\frac{M_1\sigma_{\min}^2}{2K}}
\end{align}
where $M_{1}=M-K-1$ (again slightly abusing notation used earlier)
Thus, we have,
\begin{align}
\label{eq:prob-succ-uncorrelated-dcitionary}
\mathbb{P}(E_{s}) \ge \left(1-K e^{-\frac{\sigma_{\min}^2\sqrt{M_1}}{K}}e^{-\frac{M_1\sigma_{\min}^2}{2K}}\right)^{N-K}(1-e^{-cM})
\end{align}
To complete the proof, we need to find bounds on the number of measurements $M$, that will allow recovery with high probability. 

To that end, first note that, whenever $x\ge 0,\ (1-x)^k\ge 1-kx,\ \forall k\in \mathbb{N}$. Thus, we can write, \begin{align*}
\mathbb{P}(E_S)\ge 1-K(N-K)e^{-\frac{\sigma_{\min}^2\sqrt{M_1}}{K}}e^{-\frac{M_1\sigma_{\min}^2}{2K}}-e^{-cM}
\end{align*}    
Further observe that $e^{-\frac{\sigma_{\min}^2\sqrt{M_1}}{K}}\le K/(\sigma_{\min}^2\sqrt{M_1})$, and  allowing to further simplify the probability expression as \begin{align*}
\mathbb{P}(E_S)\ge 1-\frac{K^2(N-K)}{\sigma_{\min}^2\sqrt{M_1}}e^{-\frac{M_1\sigma_{\min}^2}{2K}}-e^{-cM}
\end{align*}
Now, use assumptions, $M>2K,\ N>K\sqrt{K}$, and the simple fact that $K^2(N-K)<N^3$, to get \begin{align*}
\mathbb{P}(E_S) & \ge 1-\frac{N^3}{\sigma_{\min^2}\sqrt{K}}e^{-\frac{M_1\sigma_{\min}^2}{2K}}-e^{-cM}\\
\ & \ge 1-1/\sigma_{\min}^2\left(\frac{N}{K}\right)^8e^{-\frac{M_1\sigma_{\min}^2}{2K}}-e^{-cM}
\end{align*}
Now, the third term can be absorbed into the second, probably by some change in constants, to produce the following simplified expression, \begin{align*}
\mathbb{P}(E_S)\ge 1-c_1\left(\frac{N}{K}\right)^8e^{-\frac{M_1\sigma_{\min}^2}{2K}}
\end{align*}
It takes little effort to see now that the failure probability can be upper bounded by $\delta$ where $\delta$ is some small constant $\delta\in (0,1)$, if $M_1>CK\ln (N/({c_2(\delta) K}))$, where $C$ is some suitably chosen constant, and $c_2$ is some constant that depends only on $\delta$.

\end{proof}
The process of choosing the constant $C$ in the proof of Lemma.~\ref{lem:least-singular-value-lem3} can be made more rigorous to get rough estimates of $C$. See Appendix.~\ref{sec:constant-C-lem:uncorrelated-prob-bound} for details.
\section{Experiments}
Several sets of numerical experiments are carried out to verify the claims presented in this paper.
In the first experiment we verify the result of Theorem~\ref{thm:uncorrelated-dcitionary-recovery-probability} by an experiment in which we empirically calculated the number of measurements required by the OLS algorithm to recover a sparse signal of dimension $N$ with a probability equal to say $0.95$.  The experiment was performed with for $N=1024$ and $N=3000$. Figure \ref{fig:measurement_vs_sparsity} shows the accuracy of the estimates. The solid line in each of the figures is drawn after estimation of the constant $C$. We can easily see that the solid line matches reasonably well to the actual data points. However, it is worth sating here that our calculation of the associated constant $C$ (see Appendix.~\ref{sec:constant-C-lem:uncorrelated-prob-bound} is too high as compared to the constant that is observed empirically. Nevertheless, we are able to show that number of measurements required required by OLS is of the order of $K\log(N/K)$ which is an improved result as compared to the previous results \cite{tropp2007signal} for greedy algorithms like OMP which show that measurements required are of the order of $K\log N$.

In the second experiment we compare the numerical simulation results for percentage of signals recovered vs number of measurements $M$ for $N=1024$, against theoretical estimates of lower bounds of probability of success, as found in the proof of Theorem.~\ref{thm:uncorrelated-dcitionary-recovery-probability}. From the plots in Figure.~\ref{fig:prob-recovery-vs-measurements} we see that though the shape of the lower bound on recovery probability curve, as estimated from theory, matches well with the simulations, the actual values have a large gap. We attribute the large gap between the empirical and the theoretical values to the analysis technique and the constants produced thereof. We admit that this is an intrinsic limitation of our analysis approach, which was also the case for Tropp's result on OMP~\cite{tropp2007signal}, where he listed a few such limitations and possible reasons of those arriving from the analysis technique.

Finally, the last experiment is done to verify the runtime complexity of OLS, as claimed in Section.~\ref{sec:OLS-runtime-complexity} of the algorithm as $\mathcal{O}(KMN)$. A clever modification to the implementation of OLS algorithm as suggested previously makes the algorithm run in linear time wrt. to $K,M$ and $N$. The experimental results in Fig.~\ref{fig:running_complexity_ols} shows the correspondence between theory and experiment.
\label{sec:experiments-uncorrelated-dictionary}
\begin{figure}[t!]
\centering
  \includegraphics[height=2.5in,width=3.5in]{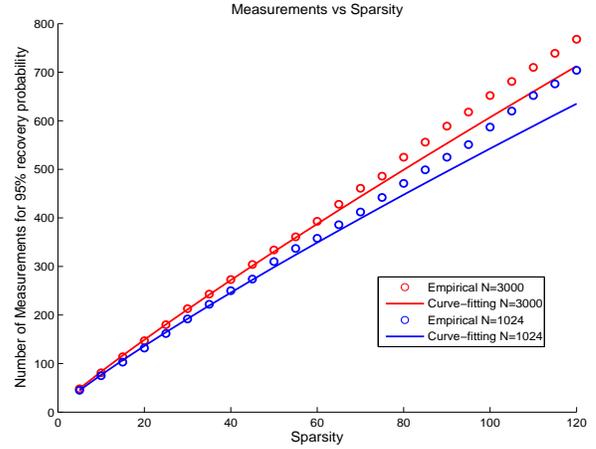}
  \caption{Measurements for 95 percent recovery probability with $N=3000$}
  \label{fig:measurement_vs_sparsity}
\end{figure}
\begin{figure}[ht!]
\centering
  \includegraphics[height=2.5in,width=3.5in]{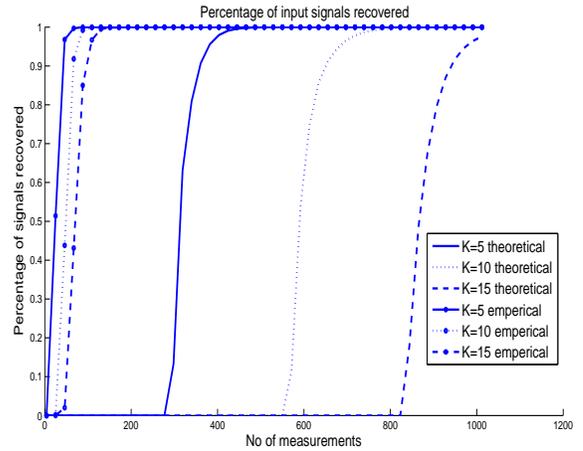}
  \caption{Probability of recovery vs no. of measurements}
  \label{fig:prob-recovery-vs-measurements}
\end{figure}
\begin{figure}[ht!]
\centering
  \includegraphics[height=2.5in,width=3.5in]{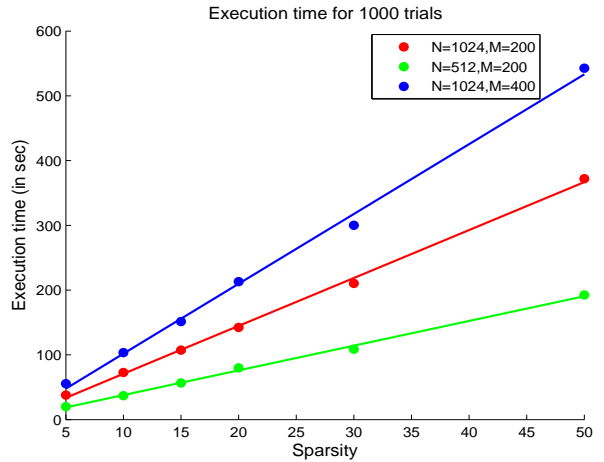}
  \caption{Execution time in seconds for OLS}
  \label{fig:running_complexity_ols}
\end{figure}
\section{Signal recovery using OLS with Hybrid Dictionaries}
In the preceding sections it was shown, both theoretically and empirically that, Orthogonal Least Squares algorithm can efficiently recovers a sparse signal when the measurement matrix consists of i.i.d Gaussian entries. A huge amount of research work has been devoted to the sparse signal representation/recovery of signal when the associated measurement matrix/dictionary satisfies strong RIP bounds. If a matrix satisfies RIP of an order $K$, it is an indicator of the fact that every $K$ columns of the matrix are almost orthogonal to each other. However, RIP is a measure used to assist worst case analysis of recovery algorithms. Average case analysis, on the other hand, studies the performance of recovery algorithms in a Monte-Carlo setup. From these average case performance plots, It has been noted in the literature, that recovery algorithms often can perform significantly well even in the presence of sensing matrices which do not satisfy the RIP bounds established by the worst case analysis. Soussen et. al.~\cite{soussen2013joint} has talked about this kind of matrices and has empirically shown that OLS is a better recovery algorithm than OMP for these kind of matrices. The measurement matrix $\mathbf{\Phi_{i}}$ has entries $\bm{\phi_{i}} = u_{i}\bm{1}+\bm{n}$ where $\bm{\phi_{i}}$ represents the $i^{th}$ column of the matrix, $u_{i}$ represents a uniformly distributed random variable from 0 to T and $\mathbf{n}$ is a vector whose entries are i.i.d gaussian distributed. Thus,$\mathbf{\Phi} \in \mathbb{R}^{M \times K}$
\begin{equation}
\mathbf{\Phi}=\mathbf{1}\mathbf{u}^{T}+\mathbf{N}
\end{equation}
where $\mathbf{u}=\lbrace u_{i} \rbrace_{i=1}^{K}$ is vector of i.i.d uniformly distributed random variables while $N$ is matrix containing i.i.d Gaussian entries. However, there is nothing special about the vector $\mathbf{1}$, we can choose any other vector in place of it.\\

It is important to note here that before using a greedy algorithm like OLS or OMP to recover sparse signal from hybrid dictionary, it is mandatory to normalize the columns since norm varies highly from column to column unlike in the case of gaussian dictionaries where the norm of a column is highly concentrated around the mean. Since OMP/OLS relies on the criteria of using inner product to select a column, a column with a large norm will be highly probable to be selected even though it is not a part of support set. So, in a normalized Hybrid matrix every column $\bm{\phi_{i}}=\alpha_{i}(u\bm{1}+\bm{n})$ where $\alpha_{i}=1/\|u\bm{1}+\bm{n} \|_{2}$. For signal recovery, a normalized Hybrid matrix is used to identify the columns in the support set and then original hybrid matrix is used to calculation signal $\bm{x}$

\subsection{Motivation of studying signal recovery with Hybrid dictionaries}
Soussen et. al. \cite{soussen2013joint} have extended the idea of Tropp's \cite{tropp2004greed} Exact Recovery Condition (ERC) on OMP to ERC on OLS. They found the ERC at any step for OMP and OLS. If ERC condition holds true at a certain iteration, it can be shown that the algorithm will be able to successfully recovery the rest of the columns without any uncertainty i.e with Probability 1.\\

For Gaussian dictionaries Soussen et. al. have empirically shown that the probability ERC is met at any iteration is same for both OMP and OLS but for Hybrid dictionaries, ERC can be guaranteed to hold with a high probability at very early iteration for OLS than for OMP. In other words that likelihood that OLS identifies correct column from the support set increases with every iteration.

Though, the normalized Hybrid measurement matrix as defined above consists of independently generated random columns but they are structured. Each and every column of the matrix is distributed around a single vector which is $\bm{1}$ in our case. Secondly, it can be shown the smallest singular value of this matrix is indeed very close to zero with high probability i.e we can find a vector $u$ such that $\mathbf{\Phi}\bm{u} \approx 0$.

As the probability of identifying correct column from the support set increases with every iteration for OLS, later in this paper, we'll show how OLS can be utilized to recover correct support set by simply running the algorithm for more than $K$ number of iterations.

\subsection{A note on Soussens et al \cite{soussen2013joint}}
Soussen's et al \cite{soussen2013joint} have inspired us to study deeply about OLS and its recovery performance with Hybrid Measurement matrix. Some of the chief contributions of the Soussen's paper are 

\begin{itemize}
\item
Extension of Tropp's ERC-OMP to any arbitrary iteration for both OMP and OLS 
\begin{equation}
\max_{j\notin Q^{*}}F^{Oxx}_{Q^{*},Q}(\mathbf{a_{j}}) < 1 
\end{equation}
with Card(Q)=$q(<k)$ where k is the actual sparsity of the signal to be recovered.
\item
With Hybrid dictionaries,the phase transition curve for OLS was empirically shown to be significantly higher than OMP.
\end{itemize}
This establishes the fact that the once the OLS reaches a given iteration and is able to recover a proper subset of support-set, it is guaranteed that it will be able to recover complete support set.


This has inspired us to study what is so special about the OLS algorithm that it is relatively  better than OMP for recovering the sparse signal especially when the columns of the dictionaries are highly coherent as in Hybrid matrix. Also, as stated earlier, the smallest singular value of the Hybrid matrix is very close to zero with high probability.


In the following section, we have recreated the experiments to analyze the behavior of OLS and OMP for Hybrid dictionaries.In the subsequent sections we'll also provide a mathematical explanation of the experiments.

\subsection{Experiments}
Here, we will discuss about the experiments which were done to study performance of Orthogonal Least Squares and Orthogonal Matching pursuit for Hybrid dictionaries as well as Gaussian Dictionaries. We consider here a hybrid dictionary $M \times N$ where $N=256$. we choose a fixed sparsity say $K=12$. For every $M$, we perform 1000 trials where we generate random measurement matrices and use these matrices for measuring a randomly generated K(=12) sparse signal under non-noisy conditions. Then we use OLS as well as OMP algorithm to recover the signal. The probability of recovery is calculated as
\begin{equation}
\begin{split}
&\textnormal{Probability of recovery}\ = \\
&\frac{\textnormal{Percentage signal recovered successfully}}{\textnormal{Number of trials}}
\end{split}
\end{equation}
In the above experiments, we have calculated conditional success probability for both OLS and OMP with Gaussian and Hybrid Dictionaries. We know that both OLS and OMP algorithm goes through $K(=12)$ iterations. At every iteration, the algorithms select a particular column from the measurement matrix. The algorithm is said to be successful if it chooses correct column at every iteration. While doing these experiments we have empirically calculated $P(S_{i})$ for every iteration $i=1,2\cdots,\ K$ where $P(S_{i})$ where $P(S_{i})$ denotes success at all iterations from $j=1,2,3,\cdots,\ i$. Therefore, we have $P(S_{i}|S_{i-1}))=\frac{P(S_{i} \cap S_{i-1})}{P(S_{i-1})}=\frac{P(S_{i})}{P(S_{i-1})}$. Effectively, $P(S_{i}|S_{i-1})$ denotes conditional probability that the $i^{th}$ iteration is successful given that previous iterations were all successful.

One can see from the figure\ref{fig:OMP_OLS_with_Gaussian} that for Gaussian Dictionaries, both for OLS and OMP, one can see the curves $P(S_{i}|S_{i-1}) vs M $ is continuously increasing with the iteration and achieves the value 1 subsequently for Measurements closer to $N$(dimension of the signal).The final figure shows the overall recovery performance for both OMP and OLS. We can see OLS is superior to OMP in terms of recovery probability but the difference is not appreciable.\\

While, for the figure \ref{fig:OMP_OLS_with_Hybrid}, which is same experiment as above but for Hybrid Dictionaries, one can see that curve$P(S_{i}|S_{i-1}) vs M$ shows an increasing trend with increasing $i$ only for OLS algorithm while it is not the case with the OMP algorithm. Going by this trend one can anticipate that there should exist an iteration after which OLS successfully recovers all columns from true support set with probability 1. We found this to be true but only for OLS while we were unable to find any such step/iteration after which OMP is successful with probability 1. In simple words, the implication of this observation is that there exist an iteration (say j) that once OLS is successful at all the iterations prior to j it is guaranteed that it will be successful at all the subsequent iterations i.e $j,j+1,\cdots,\ K$ (with probability 1) while no such iteration exists for OMP. Infact one can show analytically that if OLS is successful till $K-1$ iterations, it will definitely 
be successful at the last iteration i.e. $K_{th}$ iteration. 


\begin{figure}[t!]
\centering
\begin{subfigure}{0.5\textwidth}
  \centering
  \includegraphics[height=2.5in,width=3.5in]{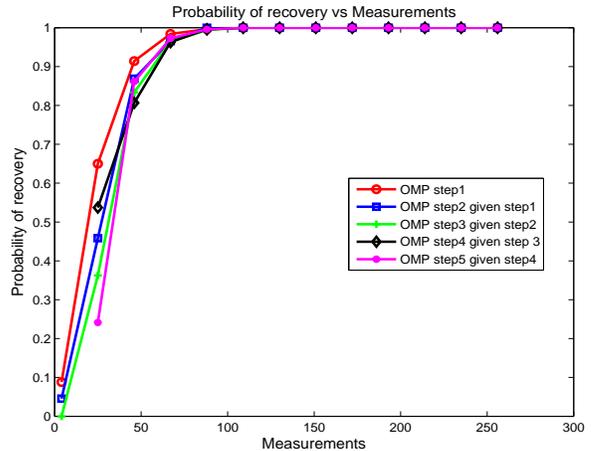}
  \caption{OMP with Gaussian Dictionary}
  \label{fig:Gaussian_OMP}
\end{subfigure}%
\hfill
\begin{subfigure}{0.5\textwidth}
  \centering
  \includegraphics[height=2.5in,width=3.5in]{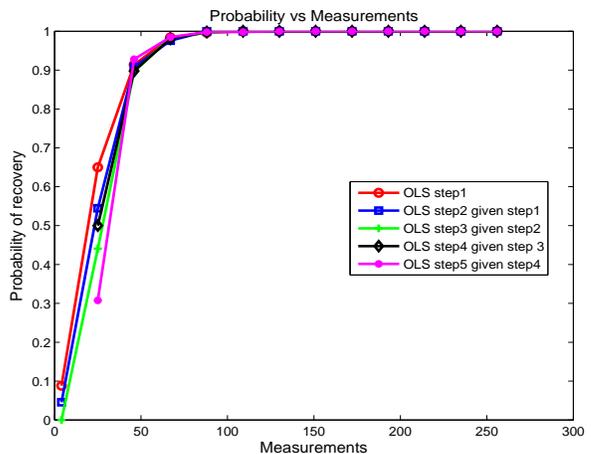}
  \caption{OLS with Gaussian Dictionary}
  \label{fig:Gaussian_OLS}
\end{subfigure}
\caption{Conditional success Probability for OMP and OLS with Gaussian Dictionary}
\label{fig:OMP_OLS_with_Hybrid}
\end{figure}

\begin{figure}[t!]
\centering
\begin{subfigure}{0.5\textwidth}
  \centering
  \includegraphics[height=2.5in,width=3.5in]{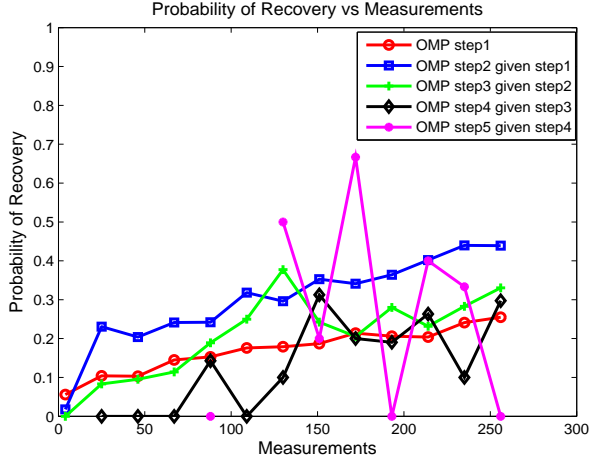}
  \caption{OMP with Hybrid Dictionary, T=100}
  \label{fig:Hybrid_OMP}
\end{subfigure}%
\hfill
\begin{subfigure}{0.5\textwidth}
  \centering
  \includegraphics[height=2.5in,width=3.5in]{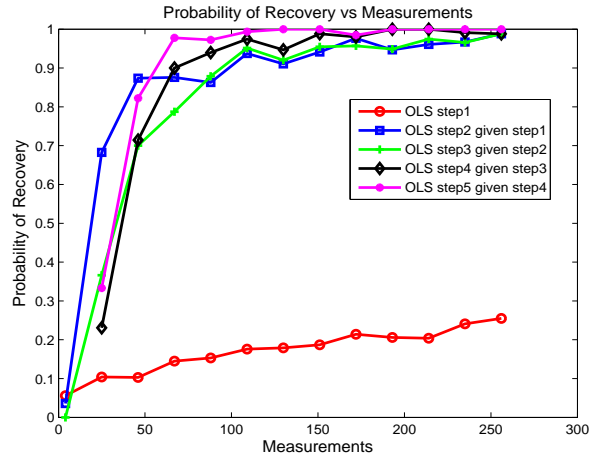}
  \caption{OLS with Hybrid Dictionary, T=100}
  \label{fig:Hybrid_OLS}
\end{subfigure}
\caption{Conditional success Probability for OMP and OLS with Hybrid Dictionary}
\label{fig:OMP_OLS_with_Gaussian}
\end{figure}


\subsection{Aanlytical justification for the empirically observed phenomenon}
\label{sec:ols-hybrid-dictionary}
In this section, we provide an explanation of the phenomenon observed in the above set of experiments. We have observed in the above experiments with Hybrid dictionary probability of success of $2^{nd}$ iteration and beyond (in OLS) conditioned on the event that the previous iteration is successful continuously increases which is not the case with OMP. Also, there exist an iteration such that if OLS is successful till that iteration, it is guaranteed to be successful for the subsequent iterations thus recovering the entire support set successfully.

We will give the justification of these empirically oberved phenomenon with respect to the generalized hybrid dictionary, defined in Section.~\ref{sec:random-measurement-ensemble}. This is achieved by first stating a series of lemmas, that will be needed to give the proof of our main theorem on the performance of OLS in generalized hybrid dictionaries.

Recall that $W_r=\spn{\bvec{a}_1,\cdots,\ \bvec{a}_r}$. A natural question arises, whether there is a collection of $r$ such columns of the generalized hybrid dictionary, that can span $W_r$ with high probability. The following lemma shows that this indeed is the case, though the probability of happening depends highly on the level of correlation, which is a function of $T$:
\begin{lem}
\label{lem:prob_hybrid_vector_align}
Let $\{\bm{\phi}_i/\norm{\bm{\phi}_i}\}_{i=1}^N$ be a collection of $N$ normalized columns forming a generalized hybrid dictionary, with orthonormal basis $\{\bvec{a}_i\}_{i=1}^r$, and parameter $T>0$ (see Definition.~\ref{defn:hybrid}). Define the events $$A_{ij}:=\left\{\left|\inprod{\frac{\boldsymbol{\phi}_i}{\norm{\boldsymbol{\phi}_i}}}{\bvec{a}_j}\right|\ge 1-\delta\right\}$$ $\forall\ 1\le  i\le N,\ 1\le j\le r$. Then, 
 \begin{align}
\mathbb{P}\left(A_{ij}\right)\ge :p(\delta)\ \forall i,j
\end{align}
where \begin{align}
\lefteqn{p(\delta)}& &\\
\ & = \sup_{\tiny{\sigma>M-1+(r-1)T^2} }2\left\{\left(1-e^{g(\sigma)/2}\right)\mathbb{E}_u\left(Q\left(\sqrt{M}(\sqrt{\sigma}\delta_1 -u)\right)\right)\right\}
\end{align}
where $u\sim\mathcal{U}[0,T),\ \delta_1^2=\frac{(1-\delta)^2}{1-(1-\delta)^2}$, 
\begin{align*}
g(\sigma)&=-\sigma-(r-1)T^2+\left(h(\sigma)-(M-1)\ln\left(\frac{M-1+h(\sigma)}{2\sigma}\right)\right)
\end{align*}
and \begin{align*}
h(\sigma)=\sqrt{(M-1)^2+4\sigma(r-1)T^2}
\end{align*}
\end{lem}
\begin{proof}
The proof is postponed to Appendix.~\ref{sec:proof_lem_hybrid_vector_align}
\end{proof}
\begin{lem}
\label{lem:hybrid_subspace_align}
 Let $L$ be a positive integer such that $r\le L\le N$. Then it follows that, whenever $\delta \in (0,0.293)$, and $p(\delta)\in [0,1/r]$,\begin{align}
 &\mathbb{P}\left(\mbox{\emph{in every collection of }}L \ \mbox{\emph{columns, }}\right.\nonumber\\
 &\left.\exists \mbox{\emph{at least one set of $r$ columns, indexed by}}\right. \\
 &\left.\ i_1,i_2,\cdots,\ i_r, \ \mbox{\emph{such that the event}} \right.\nonumber\\
&  \left.\ A_{i_11}\cap A_{i_22}\cap\cdots\cap A_{i_rr}\ \mbox{\emph{takes place}}\right)\nonumber\\
& \ge \sum_{j=0}^r (-1)^j \binom{r}{j}(1-jp(\delta))^L
 \end{align}
\end{lem} 
\begin{proof}
The proof is postponed to Appendix.~\ref{sec:proof_lem_hybrid_subspace_align}. 
\end{proof}
\begin{lem}
\label{lem:proj_error_bound}
Let, w.l.o.g., the support selected upto the $r^{\mathrm{th}}$ step of OLS  be $T_r=\{1,\ 2,\cdots,\ r \}$, such that $\min_{1\le k\le r}\abs{\inprod{\bm{\phi}_{k}/\norm{\bm{\phi}_k}}{\bvec{a}_k}}\ge 1-\delta$, then, 
\begin{align}
\norm{\dualproj{T_r}\bvec{a}_i}\le \sqrt{2\delta-\delta^2}
\end{align}
for all $1\le i\le r$.
\end{lem}
\begin{proof}
The proof follows from the simple observation that \begin{align*}
\norm{\dualproj{T_r}\bvec{a}_i} & \le\norm{\dualproj{\{i\}}\bvec{a}_i} \\
\ &=\sqrt{1-\frac{\abs{\inprod{\bm{\phi}_i}{\bvec{a}_i}}^2}{\norm{\bm{\phi}_i}^2}}\\
\ &\le \sqrt{1-(1-\delta)^2}=\sqrt{2\delta-\delta^2}.
\end{align*}
\end{proof}
 
The following lemma shows that the generalized hybrid dictionary defined above is indeed correlated in the sense that it has a high worst case coherence.
\begin{lem}
\label{lem:coherence-generalized-hybrid-dictionary}
For a sensing matrix $\bm{\Phi}\in \real^{M\times N}$, with columns belonging to a generalized hybrid dictionary as defined in~\ref{defn:hybrid}, for some $\delta\in (0,0.293)$, \begin{align*}
\lefteqn{\mathbb{P}\left(\mu(\Phi)\ge 1-4\delta+2\delta^2\right)} & & \\
\ & \ge \sum_{k=0}^r (-1)^{k}\binom{r}{k}\sum_{j=0}^{k} \binom{k}{j} \frac{(2r)!}{(2r-j)!}p(\delta)^j(1-kp(\delta))^{2r-j} 
\end{align*}
where $p(\delta)$ was defined in Lemma.~\ref{lem:prob_hybrid_vector_align}. 
\end{lem}
\begin{proof}
The proof is postponed to Appendix.~\ref{sec:proof-lem-coherence-generalized-hybrid-dictionary}.
\end{proof}
It is evident that the poor performance of OLS/OMP at the very first iteration is due the presence of a constant bias term residing in the space $W_r$, being added to each of the columns of the Hybrid measurement matrix. This makes the columns of the matrix to be correlated or packed closely to each other. This actually results in incorrectly identifying the columns since all the columns vectors are packed very closely to each other.\\  
The above lemma will be helpful in proving our claim that from $2^{nd}$ iteration onwards in OLS, the strength of the bias term added due to vector in the space $W_r$ decreases due to which the correlation among the columns decreases. This helps the algorithm to easily figure out the correct column from the support set.

Our aim was to explain the reason behind the improved conditional recovery performance of OLS from second iteration onwards. The following theorem and the discussion which follows will provide an explanation for the phenomena observed in the experiments.
\begin{thm}
\label{thm:ols-acting-like-omp}
Let $T^k$ be the set of indices selected in the first $k(r\le k<K)$ iterations of OLS, and, w.l.o.g., assume that $T^k=\{1,2,\cdots,\ k\}$. Also, assume that the first $k$ iterations of OLS are successful, that is $T^k\subset T $, where $T$ is the actual (unknown) support of the unknown vector $\bvec{x}$. Then, at the $(k+1)^{\mathrm{th}}$ iteration, the probability that OLS chooses another correct index, is at least as large as the probability of OMP choosing a correct index in its first iteration, with a $(K-k)$-sparse unknown vector, and with a hybrid sensing matrix, with non-orthonormal bias, such that the Frobenius norm of the bias matrix is upper bounded by $\sqrt{r}\kappa(\delta)$, where $\kappa(\delta):=\sqrt{2\delta-\delta^2}$.
\end{thm}
\begin{proof}
Since the unknown vector $\bvec{x}$ is $K$-sparse with support set $T$, (w.l.o.g. assumed to be $\{1,2,\cdots,\ K\}$) the measurement vector $\bvec{y}$ can be represented as \begin{align*}
\bvec{y}=x_1\bm{\phi}_1/\norm{\bm{\phi}_1}+\cdots+x_K\bm{\phi}_K/\norm{\bm{\phi}_K}
\end{align*}

Now, note that the residual after $k^{\mathrm{th}}$ iteration becomes, \begin{align*}
\bvec{r}^k=\dualproj{T^k}\bm{y}=x_{k+1}\dualproj{T^k}\frac{\bm{\phi}_{k+1}}{\norm{\bm{\phi}_{k+1}}}+\cdots+x_K\dualproj{T^k}\frac{\bm{\phi}_K}{\norm{\bm{\phi}_K}}
\end{align*}
Observe that the operator $\dualproj{T^k}$ has two distinct eigenvalues, $0$ and $1$. Let the corresponding eigenspaces have dimensions $d_0$ and $d_1$ respectively. Then, it follows that $\dim(R(\bvec{\Phi}_{T^k}))=d_0$ and $d_0+d_1=M$. One can construct a unitary matrix $\bvec{U}$ with the columns as the orthonormal eigenvectors corresponding to the eigenvalues of $\dualproj{T^k}$, so as to get, \begin{align*}
\bvec{U}=[\bvec{c}_1\ \cdots\ \bvec{c}_{d_1}\ \bvec{b}_1\ \cdots\ \bvec{b}_{d_0}] 
\end{align*}
where $\{\bvec{b}_1,\cdots,\ \bvec{b}_{d_0}\}$ and $\{\bvec{c}_1,\ \cdots,\ \bvec{c}_{d_1}\}$ are a set of orthonormal bases for the eigenspaces corresponding to the $0$ and $1$ eigenvalues, respectively, of $\dualproj{T^k}$. Denote $\bvec{U}_0=[\bvec{b}_1\ \bvec{b}_{2}\ \cdots\ \bvec{b}_{d_0}]$, and $\bvec{U}_1=[\bvec{c}_1\ \cdots\ \bvec{c}_{d_1}]$. Then, note that $\bvec{U}_0,\ \bvec{U}_1$ are not square, but they satisfy $\bvec{U}^T_0\bvec{U}_0=\bvec{I}_{d_0}$, $\bvec{U}_1^T\bvec{U}_1=\bvec{I}_{d_1}$.

Recall that $\bm{\phi}_i:={\sum_{j=1}^r u_{ij}\bvec{a}_j+\bvec{n}_i}$
Take some $k+1\le i\le K$. Then,  \begin{align*}
\dualproj{T^k}\frac{\bm{\phi}_i}{\norm{\bm{\phi}_i}}=&\frac{\sum_{j=1}^r u_{ij}\dualproj{T^k}\bvec{a}_j+\dualproj{T^k}\bvec{n}_i}{\norm{\sum_{j=1}^r u_{ij}\bvec{a}_j+\bvec{n}_i}}
\end{align*}
Now, since $r\le k<K<N$, Lemma.~\ref{lem:hybrid_subspace_align} dictates that, by appropriately choosing some $\delta\in [0,0.293]$, such that $p(\delta)<1/r$ ($p(\delta)$ was defined in Lemma.~\ref{lem:prob_hybrid_vector_align}), w.h.p., one can find $r$ indices, w.l.o.g., taken as $\{1,2,\cdots,\ r\}$, from the collection $\{1,2,\cdots,\ k\}$ such that $\abs{\inprod{\bm{\phi}_i/\norm{\bm{\phi}_i}}{\bvec{a}_i}}\ge 1-\delta$, where $\{\bvec{a}_j\}_{j=1}^r$ are defined in Definition~\ref{defn:hybrid} of the generalized hybrid dictionary.
Then, from Lemma.~\ref{lem:proj_error_bound}, we have that for any $j,\ 1\le j\le r$, \begin{align*}
\norm{\dualproj{T^k}\bvec{a}_j}\le & \norm{\dualproj{T_r}\bvec{a}_j} \le \kappa(\delta)
\end{align*}
where $T_r:=\{1,2,\cdots,\ r\}$.

On the other hand, note that $\dualproj{T^k}=\bvec{U}\bvec{\Sigma}{\bvec{U}}^T$, where $\bvec{\Sigma}=\mathrm{diag}\left(1,\ 1,\ \cdots\ ,0,\ 0,\ \cdots,\ 0\right)$, with $d_0$ number of $0$'s and $d_1$ number of $1$'s. Thus, $\dualproj{T^k}=\bvec{U}_1\bvec{U}_1^T$. Consequently, for any $1\le j\le r$, $\dualproj{T^k}\bvec{a}_j=\bvec{U}_1\bm{\epsilon}_j$, where $\bm{\epsilon}_j=\bvec{U}_1^T\bvec{a}_j$. Observe that $\norm{\bm{\epsilon}_j}=\norm{\bvec{U}_1\bm{\epsilon}_j}\le \kappa(\delta)$.

Also note that, $\dualproj{T^k}\bvec{n}_i=\bvec{U}_1\tilde{\bvec{n}_i}$, where $\tilde{\bvec{n}_i}=\bvec{U}_1^T\bvec{n}_i$. Now, \emph{given} the columns $\{\bm{\phi}_i/\norm{\bm{\phi}_i}\}_{i\in T^k}$, $\tilde{\bvec{n}_i}\sim\mathcal{N}(\bvec{0},M^{-1}\bvec{I}_{d_1})$.

Putting everything together, we get \begin{align*}
\dualproj{T^k}\frac{\bm{\phi}_i}{\norm{\bm{\phi}_i}}&=\frac{\bvec{U}_1\left(\sum_{j=1}^r u_{ij}\bm{\epsilon}_j+{\tilde{\bvec{n}}_i}\right)}{\norm{\sum_{j=1}^r u_{ij}\bm{\epsilon}_j+\tilde{\bvec{n}}_i}}
\end{align*} 
where $\tilde{\bm{\phi}_i}=\sum_{j=1}^r u_{ij} \bm{\epsilon}_j+\tilde{\bvec{n}}_i$. So, in the $(k+1)^{\mathrm{th}}$ step of OLS, the new index is chosen by finding the index that maximizes 
$\abs{\inprod{\tilde{\bm{\phi}_i}}{\bvec{r}^k}}/\norm{\tilde{\bm{\phi}_i}}$ where $\bvec{r}^k$ is obtained from measuring a $K-k$ sparse vector, and $\left\{\tilde{\bm{\phi}_i}/\norm{\tilde{\bm{\phi}_i}}\right\}$ form a dictionary such that $\tilde{\bvec{\phi}_i}=\sum_{j=1}^r {u}_{ij}{\bm{\epsilon}}_j+\tilde{\bvec{n}}_i$. Note that the bias matrix here is $\bvec{E}$, where $\bvec{E}:=[\bm{\epsilon}_1\ \bm{\epsilon}_2\ \cdots\ \bm{\epsilon}_r]$. Since $\norm{\bm{\epsilon}_j}\le \kappa(\delta)$, it is an easy matter to check that the Frobenius norm of the bias matrix is bounded by $\sqrt{r}\kappa(\delta)$.

\end{proof}
\begin{remark}
This is analogous to the generalized hybrid dictionary defined before, however, with the major difference that the columns $\{\tilde{\bm{\epsilon}}_{i}\}$ are \emph{not} orthonormal. Though we can see that the Frobenius norm of the bias can decrease with suitable choice of $\delta$. This is further clarified in the corollary below that stresses on the specific case with $r=1$.
\end{remark}
\begin{cor}
\label{cor:ols_hybrid_dictionary_r=1_case}
For the case when $r=1$, the algorithm OLS, after picking $k(<K)$ correct indices in the first $k$ iterations, acts as OMP with generalized hybrid dictionary with parameters $r=1$ and $T\sqrt{2\delta-\delta^2}$, in the $(k+1)^{\mathrm{th}}$ iteration.
\end{cor}
\begin{proof}
This is a straightforward implication of Theorem.~\ref{thm:ols-acting-like-omp}, for $r=1$.
\end{proof}
\begin{remark}
The importance of the implication of this corollary is far greater than the corollary itself. What it says is that after OLS is able to capture $k$ correct indices in $k$ iterations, it acts like OMP in the $(k+1)^{\mathrm{th}}$ iteration, with a hybrid dictionary which is less coherent. This in turn, acts as a warms-start for a new OLS algorithm with less coherent dictionary, and after a few iterations, if it chooses a few correct indices, it can again act as OMP with a hybrid dictionary where the coherence of the dictionary is further reduced, as the parameter decreases to $T(2\delta-\delta^2)$. This ``decorrelation'' effect is the key to the success of OLS in coherent dictionaries. As can be seen from the proof of Theorem.~\ref{thm:ols-acting-like-omp}, this feature of OLS is attributed to the selection criteria of OLS, which basically searches indices based on the angles between residual error vector and the normalized orthogonal projection error of the columns of the sensing matrix, after being projected onto the space spanned by the columns already selected. In OMP, the normalization of the orthogonal projection error vector is absent, and that is the reason for the failure of OMP in hybrid dictionaries.      
\end{remark}
\section{OLS support recovery with Hybrid dictionary}
\begin{figure}[t!]
\centering
  \centering
  \includegraphics[height=2.5in, width=3.5in]{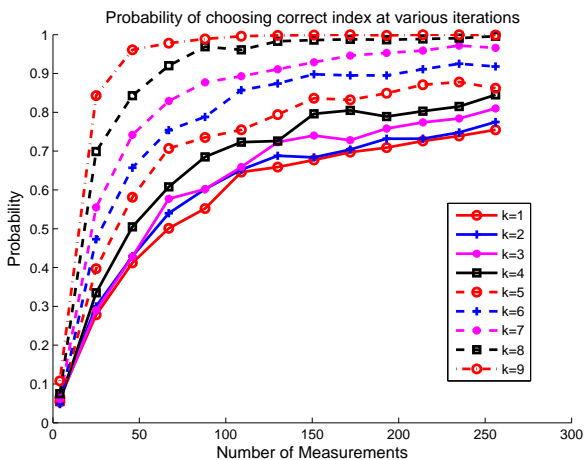}
  \caption{Prob. of choosing correct index at the start of various iteration}
  \label{fig:Prob_correct_index_iteration}
\end{figure}%
\begin{figure}[t!]
\centering
  \centering
  \includegraphics[height=2.5in,width=3.5in]{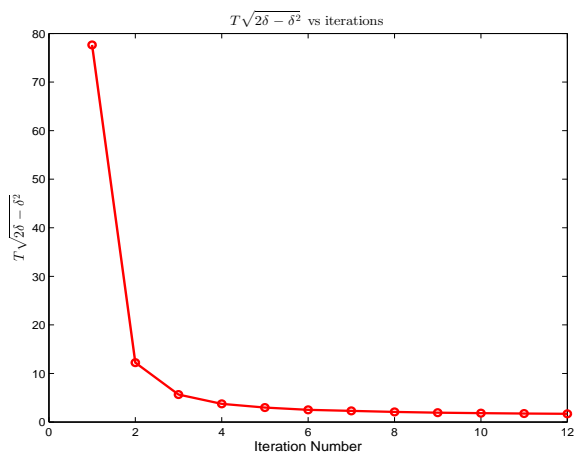}
  \caption{T$\sqrt{\delta-\delta^2}$ for different iterations}
  \label{fig:T_delta_iterations}
\end{figure}%
\begin{figure}[t!]
  \centering
  \includegraphics[height=2.5in,width=3.5in]{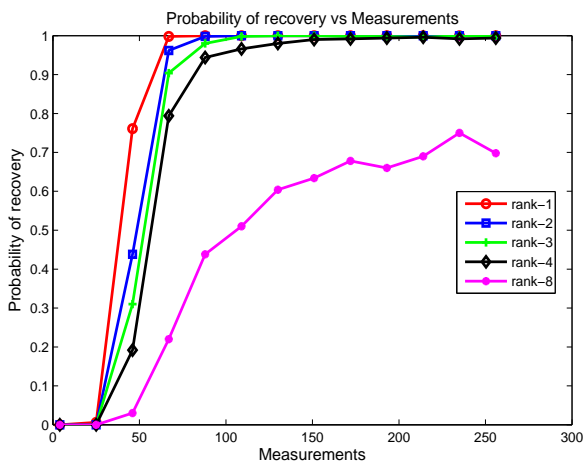}
  \caption{Probability of success for different ranks (2K iterations)}
  \label{fig:Different_ranks}
\end{figure}%
The Corollary.~\ref{cor:ols_hybrid_dictionary_r=1_case} states that, in generalized hybrid dictionary, with $r=1$, after $k(>1)$ successful iterations of OLS, the probability of success in the next will be equal to the success of the first iteration of when OMP is used to recover $K-k$ sparse signal with a modified hybrid dictionary having parameter $T\sqrt{2\delta-\delta^2}$. Now, we perform an experiment where we initially choose $T=100$ and $K=12$. For every iteration $k$, we experimentally find the value of $\delta$ using equation as in Lemmas~\ref{lem:prob_hybrid_vector_align}, and \ref{lem:hybrid_subspace_align}, for the case $r=1$, such that the probability $p(\delta)$, as defined in Lemma.~\ref{lem:prob_hybrid_vector_align} calculates out to be more than $0.99$. We plot that probability in figure \ref{fig:T_delta_iterations}. We can see the rapid decrease in the value $T\sqrt{2\delta-\delta^2}$ with the iteration numbers, as estimated in Corollary.~\ref{cor:ols_hybrid_dictionary_r=1_case}.\\

Using the values of $T\sqrt{2\delta-\delta^2}$ as in previous experiment, we perform an experiment with OMP and plot the probability of success at various iterations w.r.t. $M$ as shown in the fig. \ref{fig:Prob_correct_index_iteration}. We can see that as iteration proceeds, the probability of success increases. This justifies the phenomenon that was observed before empirically.

As we have observed that with Hybrid dictionaries, the high probability of success of OLS is assured only during last iteration, the overall probability of success of OLS is quite low. However, if we run the OLS for more than $K$ iterations lets say $2K$ iterations, the probability that $2K$ columns obtained by OLS will be definitely increase. We empirically show this in the figure \ref{fig:Different_ranks}. If we assume that $M>2K$, by projecting $\bvec{y}$ on $2K$ columns we can find the true $K$ columns of the support set since when $M>2K$ and column entries have been generated from a continuous probability space the event that every $2K$ columns is linearly independent occurs almost surely.
\section{Conclusion}

In this paper we have established the probabilistic recovery guarantee of Orthogonal Least Squares Algorithm with compressive measurements through Gaussian dictionaries under non-noisy conditions. Specifically we found lower bounds on the probability of success of OLS algorithm in unocrrelated dictionaries with Gaussian dictionaries and with normalized columns. We showed that OLS can be implemented in a way so that it has the same computational complexity as that of the popular OMP algorithm. We have defined certain type of correlated dictionary that we call geenralized hybrid dictionary, and have numerically demonstrated the competitive edge offered by OLS in these dictionaries, compared against OMP, in terms of recovery performance. We have given theoretical justifications for these numerical evidence and found out that the core reason behind the success of OLS in correlated dictionaries is a phenomenon of ``decorrelation'' of sensing matrix, which is unique to OLS because of the rule it uses at its identification step. Our future aim is to provide a more rigorous explanation of the improved performance of OLS algorithm at any general iteration, and probably find a way to use that to design algorithms with superior performances.

\appendices
\section{Proof of Lemma.~\ref{lem:theta_1-prob-upper-bound}}
\label{sec:proof-lemma-theta_1-prob-upper-bound}
\begin{proof}
For each $n\in \mathbb{N}$, define the function $g_n:(-\infty,0]$, such that $g_n(x)=\ln f_n(x)\ \forall x\in [0,1]$. Now, it is easy to verify that $f_n$ is continuous in $[0,1]$, and differentiable in $(0,1)$, which makes $g_n$ continuous in $[0,1]$, and differentiable in $(0,1)$. Now, note that $f_n(0)=1,\ f_n(1)=0$, and $f_n$ is decreasing in $[0,1]$. Using mean value theorem, we get, for any $x\in (0,1)$, \begin{align*}
g_n(x)= & g_n(0)+xg_n'(\theta x)\\
\ =& xg_n'(\theta x)
\end{align*} for some $\theta\in [0,1]$. Now, for any $x\in [0,1]$, $g_n'(x)=f_n'(x)/f_n(x)$. Since $x$ is positive, $g_n$ can be upper bounded by a linear function if a constant upper bound of $g'$ can be found. Let $h_n=-g_n'$, with domain $[0,1]$. We want to estimate $\inf_{x\in [0,1]}h_n(x)$. Now note that $h_n$ is continuous in $(0,1)$, and $h_n(x)<\infty,\ \forall x\in (0,1)$. Furthermore, observe that \begin{align*}
h_n(x)=&\frac{1}{2}\frac{(1-x)^{(n-1)/2}}{\sqrt{x}\int_{0}^{\sqrt{1-x}}\frac{u^n}{\sqrt{1-u^2}}du}\\
\implies \lim_{x\to 0+}h_n(x)=&\infty\\
\lim_{x\to 1-}h_n(x)=&\infty
\end{align*} i.e. $h_n$ is continuous on the domain $[0,1]$. Since $[0,1]$ is compact, $h_n$ attains its sup and inf in $[0,1]$. Moreover, $h_n(0+)=h_n(1-)=\infty\implies\ x^*:=\argmin_{x\in [0,1]}h_n(x)\in (0,1)$. $x^*$ satisfies \begin{align*}
h_n'(x^*)=&0\\
\ \implies f_n''(x^*)f_n(x^*)=(f_n'(x^*))^2
\end{align*} 
Then, \begin{align*}
h_n(x^*)=& -\frac{f_n'(x^*)}{f_n(x^*)}=-\frac{f_n''(x^*)}{f_n'(x^*)}
\end{align*}
Recalling the definition of $f_n$, we find, $\forall x\in (0,1)$, \begin{align*}
f_n'(x)=& -\frac{1}{2A_n}\frac{(1-x)^{(n-1)/2}}{\sqrt{x}}\\
f_n''(x)=& \frac{(1-x)^{(n-3)/2}}{4A_nx^{3/2}}(1+(n-2)x)
\end{align*}
Thus \begin{align*}
h_n(x^*)=\frac{1+(n-2)x^*}{2x^*(1-x^*)}
\end{align*}
Now, let us look at the real valued function $\phi_n$, with domain $[0,1]$, defined as \begin{align*}
\phi_n(x)=\frac{1+(n-2)x}{x(1-x)}
\end{align*}
It is straightforward to see that $\phi_n$ attains its minima at $a=\frac{1}{\sqrt{n-1}+1}$. Thus $h_n(x^*)\ge \phi_n(a)=(\sqrt{n-1}+1)^2$. Hence, \begin{align*}
g_n(x)=-xh_n(\theta x)\le -m(n)x
\end{align*}
where $m(n)=(\sqrt{n-1}+1)^2$. The desired result follows immediately.
\end{proof}
\section{Rough estimate of constant $C$ in Lemma.~\ref{lem:least-singular-value-lem3}}
\label{sec:constant-C-lem:uncorrelated-prob-bound}
We had shown in the proof of Theorem~\ref{thm:uncorrelated-dcitionary-recovery-probability}
\begin{equation}
\begin{split}
\mathbb{P}(E_{s})&=1-\frac{c_{1}N^{2.5}}{\sigma \sqrt{K}}e^{\frac{-\sigma^2 M_{1}}{2K}}-e^{\frac{-\epsilon_{1}^2K}{2}}-2Ke^{\frac{-\epsilon_{2}^2K}{8}}\\
\textnormal{where} \hspace{0.1cm} \sigma &\ge\left(1-\sqrt{\frac{K}{M_{1}}}-\epsilon_{1}\sqrt{\frac{K}{M_{1}}}\right)\left(\sqrt{1-\epsilon_{2}\sqrt{\frac{K}{M_{1}}}}\right)\\
M_{1}&=M-K-1
\end{split}
\end{equation}
We claim that no of measurements required to reduce probability of failure below $3\delta$ for some $\delta \in (0,1)$ is $CK \ln(\frac{N}{\delta K})+K+1$ for some positive constant C. Choosing, $\epsilon_{1}=\left(\frac{2\ln(N/K\delta)}{K}\right)^{\frac{1}{2}}$ and $\epsilon_{2}=\left(\frac{8\ln(N/\delta)}{K}\right)^{\frac{1}{2}}$ \\
$e^{\frac{-\epsilon_{1}^2K}{2}}=\frac{K\delta}{N} \le \delta$ and $2Ke^{\frac{-\epsilon_{2}^2K}{8}}=\frac{2K\delta}{N} \le \delta$\\
Consider, $\epsilon_{2}\sqrt{K/M_{1}}=\left(\frac{8\ln(N/\delta)}{M_{1}} \right)^{\frac{1}{2}} \le\left(\frac{8}{C} \right)^{\frac{1}{2}}$ With assumption on $M_{1}$ and $K>1$\\
$\frac{\sigma^{2}M_{1}}{K}\ge\left( \sqrt{\frac{M_{1}}{K}}-1-\epsilon_{1} \right)^{2}\left(1-(\frac{8}{C})^{\frac{1}{2}} \right)$\\
$e^{-\frac{\sigma^{2}M_{1}}{K}}=\left( \frac{\delta K}{N} \right)^ {\left[ \sqrt{C} -\frac{1}{\ln(N/\delta K)} - \sqrt{\frac{2}{K}} \right]^2 \left( 1- (\frac{8}{C})^{\frac{1}{2}} \right)/2}$\\
$\frac{c_{1}N^{2.5}}{\sigma \sqrt{K}}e^{\frac{-\sigma^2 M_{1}}{2K}}= \frac{c_{1}N^{2.5}}{\sigma \sqrt{K}} \left( \frac{\delta K}{N} \right)^ {\left[ \sqrt{C} -\frac{1}{\ln(N/\delta K)} - \sqrt{\frac{2}{K}} \right]^2/2 }$ \\
Let  $\sqrt{C}=\left( \sqrt{S} + \frac{1}{\ln(N/\delta K)} + \sqrt{\frac{2}{K}} \right)$, where $S$ is some sufficiently large number such that \\
$\frac{c_{1}N^{2.5}}{\sigma \sqrt{K}}e^{\frac{-\sigma^2 M_{1}}{2K}}= \frac{c_{1}N^{2.5}}{\sigma \sqrt{K}} \le \delta $. \\
One good choice of $S$ could be 32.
So we have $C=\left( \sqrt{32}+\frac{1}{\ln(N/\delta K)} +\sqrt{\frac{2}{K}} \right)^{2}$
\section{Proof of Lemma.~\ref{lem:prob_hybrid_vector_align}}
\label{sec:proof_lem_hybrid_vector_align}
The following simple observation will be useful for the proof of Lemma.~\ref{lem:prob_hybrid_vector_align}.
\begin{lem}
\label{lem:increasing_g}
Let $g_{a,b}$ be a real valued function, parameterized by $a,b\in \real^+$, defined as \begin{align*}
g_{a,b}(x)&=\sqrt{a^2+4bx}-x\\
\ &-a\ln\left(\frac{\sqrt{a^2+4bx}+a}{2b}\right)
\end{align*}
Then, $g_{a,b}$ is a monotonically increasing function.
\end{lem}
\begin{proof}[Proof of Lemma.~\ref{lem:prob_hybrid_vector_align}]
We first note that, given the $\{u_{ij}\}$s, the random vector $\bm{\phi}_i$ is distributed as $\mathcal{N}_{M}\left(\sum_{j=1}^r u_{ij} \bvec{a}_j,\ M^{-1}\bvec{I}_M\right)$. Let us define the matrix $\bvec{U}$ such that \begin{align*}
\bvec{U}=\begin{bmatrix}
\bvec{a}_1^T\\
\bvec{a}_2^T\\
\vdots\\
\bvec{a}_r^T\\
\bvec{a}_{r+1}^T\\
\vdots\\
\bvec{a}_M^T
\end{bmatrix}
\end{align*}  where the collection $\{\bvec{a}_i\}_{i=1}^M$ forms an orthogonal basis for $\real^M$. By construction $\bvec{U}$ is unitary which implies, for any $1\le i,j\le M$, \begin{align*}
\frac{\inprod{\bm{\phi}_i}{\bvec{a}_j}}{\norm{\bm{\phi}_i}}&=\frac{\inprod{\bm{U\phi}_i}{\bvec{Ua}_j}}{\norm{\bm{U\phi}_i}}\\
\ &=\frac{\psi_{ij}}{\norm{\bm{\psi}_i}}
\end{align*}  
where $\bm{\psi}_i=\bm{U\phi_i}=[{\psi}_{i1}\ {\psi}_{i2}\cdots\ {\psi}_{iM}]^T$. Note that, by construction, $\psi_{ij}\sim\mathcal{N}(u_{ij},1/M),\ j=1,2,\cdots,\ r$ and $\psi_{ij}\sim\mathcal{N}(0,1/M),\ j=r+1,\cdots,\ M$. Then, \begin{align*}
\lefteqn{\mathbb{P}\left(\left|\frac{\inprod{\bm{\phi}_i}{\bvec{a}_j}}{\norm{\bm{\phi}_i}}\right|\ge 1-\delta\mid\{u_{ik}\}_{k=1}^r.\right)}& &\\
\ &=\mathbb{P}\left(\frac{\psi_{ij}^2}{\sum_{k=1}^r \psi_{ik}^2}\ge (1-\delta)^2\right)\\
\ &=\mathbb{P}\left(\sum_{k\ne j}\psi_{ik}^2\le \frac{\psi_{ij}^2}{\delta_1^2}\right)\\
\ &\ge \sup_{\sigma>0}\mathbb{P}\left(\sum_{k\ne j}\psi_{ik}^2\le \sigma\right)\mathbb{P}(\psi_{ij}^2\ge \delta_1^2\sigma)
\end{align*} where $\delta_1$ is defined as in Lemma.~\ref{lem:prob_hybrid_vector_align}.
We are interested in the case where $1\le j\le r$. To bound the whole probability, we bound the two product terms separately as below
\begin{itemize}
\item Note that $\psi_{ij}\sim \mathcal{N}(0,M^{-1})\implies$ $$\mathbb{P}(\psi_{ij}^2\ge \delta_1^2\sigma)=2Q\left(\sqrt{M}(\delta_1\sqrt{\sigma}-u_{ij})\right)$$
\item Using Chernoff's bound, \begin{align*}
\lefteqn{\mathbb{P}\left(\sum_{k\ne j}\psi_{ik}^2\le \sigma\right)}& &\\
\ & \ge 1-\sup_{\theta>0}e^{-\theta \sigma}\mathbb{E}\left(\exp\left[\theta\sum_{k\ne j}\psi
_{ik}^2\right]\right)\\
\ &=1-\sup_{1/2>\theta>0}e^{-\theta\sigma}\frac{\exp\left(\frac{\sum_{k\ne j}{u_{ik}^2}\theta}{1-2\theta}\right)}{(1-2\theta)^{(M-1)/2}}\\
\ &=1-\sup_{1/2>\theta>0}\exp(f(\theta,\sigma))
\end{align*}
where $f(\theta,\sigma)=\frac{\sum_{k\ne j}{u_{ik}^2}\theta}{1-2\theta}-\theta\sigma-\frac{M-1}{2}\ln(1-2\theta)$
\end{itemize}
Then, defining $t=1-2\theta$, it is easy to observe that, for a fixed $\sigma$, $f(\theta,\sigma)$ is maximized at $$t=t_{max}=\frac{M-1+\sqrt{(M-1)^2+4\lambda_{ij}\sigma}}{2\sigma}$$ where $\lambda_{ij}:=\sum_{k\ne j}u_{ik}^2$. Note that the condition $\theta\in [0,1/2)$ constraints the range of $\sigma$ to be $[M-1+\lambda_{ij},\infty)$. Now, it is a simple matter of computation to show that $f(\theta_{\mathrm{max}},\sigma)=\left(-\sigma+g_{M-1,\sigma}(\lambda_{ij})\right)/2$. Using Lemma.~\ref{lem:increasing_g}, we get $f(\theta_{\mathrm{max}},\sigma)\le\left(-\sigma+g_{M-1,\sigma}((r-1)T^2)\right)/2$, since $\lambda_{ij}=\sum_{k\ne j}u_{ik}^2\le (r-1)T^2$. It is important to note that we can upper bound $g_{M-1,\sigma}$ in this way only when $\sigma>M-1+(r-1)T^2$.

Then, we get \begin{align*}
\lefteqn{\mathbb{P}\left(\left|\frac{\inprod{\bm{\phi}_i}{\bvec{a}_j}}{\norm{\bm{\phi}_i}}\right|\ge 1-\delta\mid\{u_{ik}\}_{k=1}^r\right)}& &\\
\ &\ge \sup_{\tiny{\sigma>M-1+(r-1)T^2}}2\left[(1-e^{g(\sigma)/2})Q\left(\sqrt{M}(\delta_1\sqrt{\sigma}-u_{ij})\right)\right]
\end{align*}
where the function $g$ is defined as in Lemma.~\ref{lem:prob_hybrid_vector_align}. Finally taking expectation with respect to the random variables $\{u_{ik}\}_{k=1}^r$, results in the desired expression for $\mathbb{P}(A_{ij})$ as in Lemma.~\ref{lem:prob_hybrid_vector_align}. 
\end{proof}
\section{Proof of Lemma.~\ref{lem:hybrid_subspace_align}}
\label{sec:proof_lem_hybrid_subspace_align}
The following lemmas will be essential for the proof of Lemma.~\ref{lem:hybrid_subspace_align}:
\begin{lem}
\label{lem:condition_mutual_independence_Aij}
Let the events $A_{ij}$ be defined as in Lemma.~\ref{lem:prob_hybrid_vector_align} for $1\le i\le N,\ 1\le j\le r$. Then, if $\delta\in (0,1-1/\sqrt{2})$, the events $A_{ij_1}$ and $A_{ij_2}$ are mutually exclusive $\forall \  1\le i\le N$ and $1\le j_1,j_2\le M,\ j_1\ne j_2$.
\end{lem}
\emph{\remark{Lemma.~\ref{lem:condition_mutual_independence_Aij} demonstrates that the condition $\delta\in (0,1-1/\sqrt{2})$ ensures that a column of a generalized hybrid dictionary cannot be ``too close'', simultaneously, to more than one of the basis vectors $\{\bvec{a}_{i}\}$}.}
\begin{proof}
Let $\{\bvec{a}_i\}_{i=1}^M$ form an orthonormal basis for $\real^M$. Then, one can uniquely express a column $\bm{\tilde{\phi}}_i:=\bm{\phi}_i/\norm{\bm{\phi}_i}$ as \begin{align*}
\bm{\tilde{\phi}}_i=\sum_{k=1}^M \epsilon_{ik}\bvec{a}_k
\end{align*}
where $\epsilon_{ik}=\inprod{\bm{\tilde{\phi}_i}}{\bvec{a}_k}$. Since the columns are normalized, we have $\sum_{k=1}^M \epsilon_{ik}^2=1$. Hence, if one has $|\epsilon_{ij_1}|\ge 1-\delta$ for some $1\le j_1\le M$, essentially, for any other index $1\le j_2\le M$, $|\epsilon_{ij_2}|\le \sqrt{1-(1-\delta)^2}<1-\delta$ under the condition $\delta\in (0,1-1/\sqrt{2})$. This concludes the proof. 
\end{proof}
\begin{lem}
\label{lem:increasing_fn_p}
The real valued polynomial $P(L,p,r)$ defined as \begin{align*}
P(L,p,r)=\sum_{j=0}^r (-1)^j \binom{r}{j}(1-jp)^L
\end{align*}
 is a positive valued monotonically increasing function of $p$ whence $p\in [0,1/r]$ where $r\ge 1, L\ge r$.
\end{lem}
\begin{proof}
The fact that $P(L,p,r)$ is positive valued for $p\in [0,1/r], L\ge r$, will be clear from the proof of Lemma.~\ref{lem:hybrid_subspace_align} as it will be shown there that this polynomial is in fact a probability expression and hence is always non negative, and for $p\in (0,1/r)$, it is strictly positive.

To show that this polynomial is increasing in $p$, note that for $r=1$, $P(L,p,1)=1-(1-p)^L$ which is clearly increasing when $p\in [0,1]$.

For $r\ge 2$, note that, from the definition of $P$, $P(L,p,r)$ is continuous and differentiable everywhere, w.r.t. $p$. Then, \begin{align*}
\lefteqn{\frac{\partial P(L,p,r)}{\partial p}}& &\\
\ &=-L\sum_{j=0}^r j(-1)^j\binom{r}{j}(1-jp)^{L-1}\\
\ &=Lr\sum_{j=0}^{r-1}(-1)^{j}\binom{r-1}{j}(1-(j+1)p)^{L-1}\\
\ &=Lr(1-p)^{L-1}P\left(L-1,\frac{p}{1-p},r-1\right)
\end{align*}
Note that $p\in [0,1/r]\implies \frac{p}{1-p}\in[0, \frac{1}{r-1}]$, which implies, from the positivity of the polynomial $P(L,p,r)$ for $p\in [0,1/r],\ L\ge r\ge 1$, that $P(L-1,p/(1-p),r-1)$ is also positive valued for $r\ge 2,\ p\in [0,1/r]$. Thus, the polynomial $P(L,p,r)$ is monotonically increasing in $p$ for $r\ge 2$.
\end{proof}
\begin{proof}[Proof of Lemma.~\ref{lem:hybrid_subspace_align}]
Let us first discuss the model of the random experiment that generates the events $A_{ij}$. To describe the experiment, we symbolize the different vectors involved according to the following notation:
\begin{description}
\item[$c_i$] denotes the symbol for the $i^{\mathrm{th}}$ normalized random column vector $\bm{\tilde{\phi}_i}$, for $1\le i\le L$   
\item[$b_j$] denotes the symbol for the $j^{\mathrm{th}}$ vector $\bvec{a}_j$ in the orthonormal basis, for $1\le j\le r$

\end{description}
We say that an ``assignment'' of $c_i$ to $b_j$ has taken place if the event $A_{ij}$ occurs. Now, in virtue of Lemma.~\ref{lem:condition_mutual_independence_Aij}, whenever $\delta\in (0,0.293)$, the events $A_{ij}$ are pairwise mutually exclusive for any fixed $i$. Also, observe that, due to the independence of the random vectors $\bm{\tilde{\phi}_i}$, any pair of events $A_{ij}$ and $A_{kl}$ are stochastically independent, as long as $i\ne k$.
Thus the model for the random experiment can be described as below:
\begin{itemize}
\item An assignment of $c_i$ to $b_j$ occurs independently of an assignment of $c_k$ to $b_l$ whenever $i\ne k$
\item For a given $i$, $c_i$ can be assigned to at most one of $b_j$, for some $1\le j\le r$.
\end{itemize}
Each of these events occur with the same probability, say $p$ which greater than $p(\delta)$ as shown in Lemma.~\ref{lem:prob_hybrid_vector_align}. 

Having set the stage for the experiment, now Let us define, for a fixed $j$, $X_j$ as random variable denoting the number of $c_i$'s assigned to $b_j$. Then, note that $0\le X_j\le L$, for any $1\le j\le r$ and $0\le \sum_{j=1}^r X_j\le L$. The objective of Lemma.~\ref{lem:hybrid_subspace_align} then boils down to calculating the quantity $\mathbb{P}(X_1\ge 1,X_2\ge 1,\cdots,\ X_r\ge 1)$. We proceed by finding out the joint distribution of the random variables $X_1,X_2,\cdots,\ X_r$. 

Let $k_1,k_2,\cdots,\ k_r$ be natural numbers such that $0\le k_1,k_2,\cdots, k_r\le L $ and $\sum_{j=1}^r k_j\le L $. Then
\begin{align*}
\lefteqn{\mathbb{P}(X_1=k_1,X_2=k_2,\cdots,\ X_r=k_r)} & &\\
\ &=\binom{L}{k_1}p^{k_1}\cdot\binom{L-k_1}{k_2}p^{k_2}\\
\ &\cdots\binom{L-(k_1+k_2+\cdots+k_{r-1})}{k_{r}}p^{k_r}\\
\ &\cdot(1-rp)^{L-(k_1+k_2+\cdots+k_{r-1})}\\
\ &=\binom{L}{k_1\ k_2\ \cdots\ k_r\  L-\sum_{j=1}^r k_j}p^{\sum_{j=1}^r k_j}(1-rp)^{L-\sum_{j=1}^r k_j}
\end{align*}
That is the random variables $X_1,X_2,\cdots,\ X_r, (L-\sum_{j=1}^r X_j)$ are multinomial distributed with parameters $p_1=p_2=\cdots=p_{r}=p,\ p_{r+1}=1-rp$. Obviously, the condition $p\in [0,1/r]$ is necessary here. Finally, we can find the desired probability as \begin{align*}
\lefteqn{\mathbb{P}(X_1\ge 1,X_2\ge 1,\cdots,\ X_r\ge 1)} & & \\
\ &=1-\mathbb{P}\left(\{X_1=0\}\cup \{X_2=0\}\cup\cdots\cup \{X_r=0\}\right)\\
\ &=1-\sum_{k=1}^r(-1)^{k-1}\\
\ & \sum_{1\le i_1<i_2<\cdots<i_k\le r}\mathbb{P}\left(X_{i1}=0,X_{i_2}=0, \cdots, X_{i_k}=0\right)\\
\ &=1-\sum_{k=1}^r (-1)^{k-1}\binom{r}{k}(1-kp)^L\\
\ &=\sum_{k=0}^r (-1)^{k}\binom{r}{k}(1-kp)^L
\end{align*}
Now, since $p\ge p(\delta)$ from Lemma.~\ref{lem:prob_hybrid_vector_align}, use of Lemma.~\ref{lem:increasing_fn_p} concludes the proof.
\end{proof}
\section{Proof of Lemma.~\ref{lem:coherence-generalized-hybrid-dictionary}}
\label{sec:proof-lem-coherence-generalized-hybrid-dictionary}
\begin{lem}
\label{lem:increasing_fn_q}
The real valued polynomial $Q(L,p,r)$ defined as \begin{align*}
Q(L,p,r)=\sum_{k=0}^r (-1)^{k}\binom{r}{k}\sum_{j=0}^{k} \binom{k}{j} \frac{L!}{(L-j)!}p^j(1-kp)^{L-j}
\end{align*}
 is a positive valued monotonically increasing function of $p$ whence $p\in [0,1/r]$ where $r\ge 1, L\ge r$.
\end{lem}
\begin{proof}
From the proof of Lemma.~\ref{lem:coherence-generalized-hybrid-dictionary}, $Q(L,p,r)$ will be understood as a probability expression, which will imply the non-negativity of $Q(L,p,r)$.
Now observe, \begin{align*}
\lefteqn{\frac{\partial Q(L,p,r)}{\partial p}} & & \\
\ =& \sum_{k=0}^r (-1)^k \binom{r}{k}\\
\ & \sum_{j=0}^k \binom{k}{j}\frac{L!}{(L-j)!}\left(jp^{j-1}(1-kp)^{L-j}\right.\\
\ & \left.-k(L-j)p^j(1-kp)^{L-j-1}\right)
\end{align*}
After a bit of manipulation, and using the identity $\binom{k}{j}-\binom{k-1}{j}=\binom{k-1}{j-1}$, it follows that \begin{align*}
\lefteqn{\frac{\partial Q(L,p,r)}{\partial p}} & & \\
\ &= \sum_{k=1}^r (-1)^{k-1} k \binom{r}{k}\\
\ & \left(\sum_{j=1}^{k}\binom{k-1}{j-1}\frac{L!}{(L-j-1)!}p^j(1-kp)^{L-j-1}\right)\\
\ &= r\sum_{k=1}^{r}(-1)^{k-1}\binom{r-1}{k-1}\\
\ &\left(\sum_{j=0}^{k-1}\binom{k-1}{j}\frac{L!}{(L-2-j)!}p^{j+1}(1-kp)^{L-2-j}\right)\\
\ &= rL(L-1)p(1-p)^{L-2}\\
\ &\left[\sum_{k=0}^{r-1}(-1)^k\binom{r-1}{k}\sum_{j=0}^k\binom{k}{j}\frac{(L-2)!}{(L-2-j)!}q^{j}(1-kq)^{L-2-j}\right]\\
\ &=rL(L-1)p(1-p)^{L-2}Q(L-2,q,r-1)>0
\end{align*}
where $q=\frac{p}{1-p}$. This proves that $Q(L,p,r)$ is increasing in $p$, as long as $p\in (0,1/r)$.
\end{proof}
\begin{proof}[Proof of Lemma.~\ref{lem:coherence-generalized-hybrid-dictionary}]
The columns of $\bm{\Phi}$ are of the form $\frac{\bm{\phi}_i}{\norm{\bm{\phi}_i}}$, where $\bm{\phi}_i$ is of the form $\bm{\phi}_i=\sum_{j=1}^r u_{ij}\bvec{a}_j+\bvec{n}_i$, where properties of these variables can be recalled from the definition~\ref{defn:hybrid} of generalized hybrid dictionary. For the sake of simplicity we will denote $\frac{\bm{\phi}_i}{\norm{\bm{\phi}_i}}$ by $\bm{\psi}_i$. Now, let us assume that, in any collection of $K$ columns, $\exists\ 2r$ columns, w.l.o.g. $\{\bm{\psi}_i\}_{i=1}^{2r}$, such that $\min_{1\le i\le r}\{\abs{\inprod{\bm{\psi}_i}{\bvec{a}_i}}\}\ge 1-\delta$, and $\min_{r+1\le i\le 2r}\{\abs{\inprod{\bm{\psi}_i}{\bvec{a}_i}}\}\ge 1-\delta$. Let $\{\bvec{a}_i\}_{i=1}^M$ form an orthonormal basis for $\real^M$. We can uniquely expand any column $\bm{\psi}_i$ as $\bm{\psi}_i=\sum_{j=1}^M \epsilon_{ij}\bvec{a}_j$. Then, note that the assumption of existence of the $2r$ columns with aforementioned property forces the constraints, $\abs{\epsilon_{ii}}\ge 1-\delta$, for $i=1,2,\cdots,\ r$, $\abs{\epsilon_{i+r,i}}\ge 1-\delta$, for $i=1,2,\cdots,\ r$, with $\sum_{j=1}^M{\epsilon_{ij}^2}=1,\ \forall i$, so that we have $\sum_{j\ne i}\epsilon_{ij}^2\le 1-(1-\delta)^2=2\delta-\delta^2$, and , similarly, $\sum_{j\ne i}\epsilon_{i+r,j}^2\le 1-(1-\delta)^2=2\delta-\delta^2$. Thus, for any $1\le i\le r$, we have \begin{align*}
\lefteqn{\abs{\inprod{\bm{\psi}_i}{\bm{\psi}_{i+r}}}} & & \\
\ &=   \abs{\sum_{j=1}^M\epsilon_{ij}\epsilon_{i+r,j}}\\
\ &\ge \abs{\epsilon_{ii}\epsilon_{i+r,i}}-\abs{\sum_{j\ne i}\epsilon_{ij}\epsilon_{i+r,j}}\\
\ &\ge (1-\delta)^2-(1-(1-\delta)^2)\quad (\mbox{Using Cauchy-Scwartz inequality})\\
\ &=1-4\delta+2\delta^2
\end{align*}
To find the probability that there are $2r$ columns $\bm{\psi}_i$ satisfying the aforementioned condition, we recall the proof of Lemma.~\ref{lem:hybrid_subspace_align}, and recognize the desired probability as $\mathbb{P}\left(X_1\ge 2,\ X_2\ge 2,\ \cdots,\ X_r\ge 2\right)$ where $X_1,\cdots,X_r,\ (K-\sum_{i=}^r X_i)$ are multinomial distributed random variables with parameters $(p_1=p,\ p_2=p,\cdots,\ p_r=p,\ p_{r+1}=1-rp)$, where $p$ was defined as in Lemma.~\ref{lem:hybrid_subspace_align}, with $p\ge p(\delta)$. consequently, the desired probability is \begin{align*}
\lefteqn{\mathbb{P}\left(X_1\ge 2,\ X_2\ge 2,\ \cdots,\ X_r\ge 2\right)} & &\\
\ &=1-\mathbb{P}\left(\{X_1\in \{0,1\}\}\cup \{X_2\in \{0,1\}\}\cup\cdots\cup \{X_r\in\{0,1\}\}\right)\\
\ &=1-\sum_{k=1}^r(-1)^{k-1}\\
\ & \sum_{1\le i_1<i_2<\cdots<i_k\le r}\mathbb{P}\left(X_{i1}\in\{0,1\},X_{i_2}\in\{0,1\}, \cdots, X_{i_k}\in\{0,1\}\right)
\end{align*} 
With a little effort, it can be shown that \begin{align*}
\lefteqn{\mathbb{P}\left(X_{i1}\in\{0,1\},X_{i_2}\in\{0,1\}, \cdots, X_{i_k}\in\{0,1\}\right)} & & \\
\ &=\sum_{j=0}^{k} \binom{k}{j} \frac{K!}{(K-j)!}p^j(1-kp)^{K-j} 
\end{align*}
Thus, the desired probability is \begin{align*}
\ & 1-\sum_{k=1}^r (-1)^{k-1}\binom{r}{k}\sum_{j=0}^{k} \binom{k}{j} \frac{K!}{(K-j)!}p^j(1-kp)^{K-j} \\
\ &=\sum_{k=0}^r (-1)^{k}\binom{r}{k}\sum_{j=0}^{k} \binom{k}{j} \frac{K!}{(K-j)!}p^j(1-kp)^{K-j} 
\end{align*}
From Lemma.~\ref{lem:increasing_fn_q}, the preceding term can be recognized as $Q(K,p,r)$. Then, an application of Lemma.~\ref{lem:increasing_fn_q} along with the fact that $p\ge p(\delta)$ establishes that \begin{align*}
\ & \mathbb{P}\left(\exists 2r\ \mbox{columns $\{\bm{\psi}_i\}_{i=1}^{2r}$ in any collection of $K$ columns}\right.\\
\ & \left.\mbox{ such that} \min_{1\le i\le r}\min\left\{\abs{\inprod{\bm{\psi}_i}{\bvec{a}_i}},\abs{\inprod{\bm{\psi}_{i+r}}{\bvec{a}_i}}\right\}\ge 1-\delta\right)\\
\ &\ge \sum_{k=0}^r (-1)^{k}\binom{r}{k}\sum_{j=0}^{k} \binom{k}{j} \frac{K!}{(K-j)!}p(\delta)^j(1-kp(\delta))^{K-j} 
\end{align*}
Thus, 
\end{proof} 
\bibliography{ref_ols_analysis}
\end{document}